\documentclass[11pt]{article}

\usepackage[a4paper,hmargin=0.85in,vmargin=0.85in]{geometry}

\usepackage[dvips]{graphicx}
\usepackage{pstricks}
\usepackage{enumerate}
\usepackage{subfig}
\usepackage{amssymb,amsmath,latexsym}
\usepackage{amsthm}
\usepackage[plainpages=false]{hyperref}
\usepackage{pstricks-add}
\usepackage{pst-plot}
\usepackage{pst-xkey}
\usepackage{multirow}
\usepackage{slashbox}
\usepackage{dsfont}
\usepackage{pifont}
\usepackage{bm}
\usepackage{textcomp}
\usepackage{microtype}
\usepackage{url}
\usepackage{units}

\newcommand{\ep}{\varepsilon}
\newcommand{\bof}[1]{\bm{#1}}
\newcommand{\st}{\operatorname{s.t.}}
\newcommand{\Prob}{\operatorname{Pr}}
\newcommand{\tr}{\operatorname{tr}}

\def\bra#1{\mathinner{\langle{#1}|}}
\def\ket#1{\mathinner{|{#1}\rangle}}

\def\brakket#1#2#3{\mathinner{\langle{#1}|{#2}|{#3}\rangle}}

\theoremstyle{plain}

\newtheorem{lemma}{Lemma}
\newtheorem{theorem}{Theorem}
\newtheorem{corollary}{Corollary}

\theoremstyle{definition}
\newtheorem{definition}{Definition}
\newtheorem{example}{Example}

\theoremstyle{remark}
\newtheorem{remark}{Remark}
\newtheorem{protocol}{Protocol}

\title{Device-Independent Quantum Key Distribution\\ with Commuting Measurements}

\author{
Esther H\"anggi\thanks{Computer Science Department, ETH Zurich, CH-8092 Z\"urich, Switzerland}
\and
Renato Renner\thanks{Institute for Theoretical Physics, ETH Zurich, CH-8093 Z\"urich, Switzerland}
}

\begin{document}

\maketitle

\begin{abstract}
We consider quantum key distribution in the device-independent scenario, i.e., where the legitimate parties do not know (or trust) the exact specification of their 
apparatus. We show how secure key distribution can be realized against the most general attacks by a quantum adversary under the condition that measurements on different 
subsystems by the honest parties commute. 
\end{abstract}

\section{Introduction}

The security of quantum key distribution is based on the laws of physics and does not rely on any (unproven) assumption of computational hardness. It does, however, assume 
that the honest parties can control their physical devices accurately and completely. If an implementation of quantum key distribution does not meet this requirement, its 
security may be compromised. For example, the BB84 quantum key-distribution protocol~\cite{bb84}  becomes completely insecure if 
the source emits several instead of single photons or if the measurement device measures only in one 
instead of two different bases. Experimentally,
 several successful attacks making use of imperfections of the physical devices have lately been implemented (see, e.g.~\cite{makarov,lo,photonics}). 

The goal of device-independent quantum key distribution is to show the security of key distribution schemes, where 
the exact description 
of the particle source (in particular, the dimension of the Hilbert space they act on) and the exact specification 
of the measurement apparatus are unknown. 
The honest parties can only check properties of the input/output behaviour of their physical system described by 
statistical tests. 

Two approaches to achieve device-independent quantum key distribution have been investigated: the first uses the validity of quantum mechanics with all its formalism, 
while the second bases security only on the non-signalling principle, i.e., the fact that the parties cannot use their physical apparatus to send messages (in particular, 
measurements on an entangled quantum state cannot be used for message transmission). It can be shown that this later condition is strictly weaker and that there exist 
examples of systems which are secure against quantum adversaries, but  insecure (or only partially secure) in a model built on the non-signalling principle only. 
The latter can therefore lead to unnecessarily low key rates or even the impossibility to create a key in certain regimes.

\subparagraph*{Our contribution: }
We give a general security proof of device-independent quantum key distribution against the 
most general attacks by any adversary limited by quantum mechanics under the sole condition that whenever 
the key distribution protocol prescribes measurements on separate subsystems, then 
these measurements commute. This condition can, for instance, be enforced by isolating the individual subsystems 
or by performing the measurements at space-like separated points. Furthermore, it is understood that the 
legitimate parties have access to a source of randomness\footnote{It is already sufficient to have a source 
producing a small number of random bits. As shown in~\cite{colbeckphd,random}, this randomness can then be expanded using a device-independent randomness expansion protocol.} and that none of the devices leaks information 
to the environment. 
Our proof method applies to a generic class of entanglement-based quantum key-distribution protocols~\cite{ekert}. In particular, we show that a protocol similar to the one proposed originally by Ekert 
reaches an asymptotic key rate of one secure bit per channel use in the noiseless limit 
even in the device-independent scenario with commuting measurements. 

The proof method we use is based on a criterion by Navascu\'es, Pirionio and Ac\'in~\cite{npa07}, to bound the information a (quantum) adversary can have about 
the legitimate parties'  
measurement results by a semi-definite program. Our main technical contribution is to show that when the honest parties share several systems with commuting measurements, 
then this semi-definite program follows a sort of \emph{product theorem} (Theorem~\ref{th:guessprod}), i.e., an adversary cannot guess the outputs of several systems any 
better than trying to guess each output individually. The resulting security proof works for any alphabet size of the inputs and outputs to the system and does not use any 
Hilbert space formalism, only convex optimization techniques.

Our security proof can also be applied in the non-device-independent scenario, i.e., where the properties of the devices are (partially) known (or trusted), leading to 
a higher key rate. 

Our technique also implies that privacy amplification, i.e., the random hashing usually performed at the end of the protocol to turn a partially secure raw key into a fully secure key, can be replaced by a deterministic function, the XOR.

\subparagraph*{Related work: }
The problem of device-independence 
has been introduced and studied by Mayers and Yao~\cite{my}, who showed security against an adversary limited to individual attacks in the noiseless scenario. 
The same scenario but allowing for noise has been treated in~\cite{mmmo}.  A device-independent quantum key distribution scheme secure against collective attacks has been 
given in~\cite{abgs}. If the devices are memoryless, this scheme can even be shown secure against the most general attacks, using a plausible but unproven assumption, as 
shown in~\cite{mckaguephd}. 

Device-independent key distribution against adversaries only limited by the non-signalling principle has first been studied by Barrett, Hardy and Kent~\cite{bhk}. 
Key distribution schemes secure against (non-signalling) individual attacks have been proposed and analysed in~\cite{AcinGisinMasanes,AcinMassarPironio,SGBMPA}. Under the additional assumption 
that a non-signalling condition holds between all subsystems, security against the most general attacks has been proven in~\cite{lluis,eurocrypt,masanesv4}.

\subparagraph*{Outline: }
We will first introduce the mathematical framework needed to describe key distribution protocols and 
define their security (Section~\ref{sec:model}). In Section~\ref{sec:condition} we review the semi-definite criterion by Navascu\'es, Pironio and 
Ac\'in~\cite{npa07} to bound the set of quantum systems. In Section~\ref{sec:qsingle} we study how to bound the security of a single system. Our main technical result 
relates the security of a single system to the security of many systems and is given in Section~\ref{sec:several}.  Using these results, we can then give a general security 
proof for device-independent quantum key distribution. We first treat the case when the marginal systems shared by 
the honest parties are independent (Section~\ref{sec:qkd}), before 
removing this requirement in Section~\ref{sec:notindependent}. Finally in Section~\ref{sec:qprotocol}, we apply our result to a specific  protocol which is secure in the device-independent scenario.

\section{Framework}\label{sec:model}

\subsection{Systems}

We define security in the context of \emph{random systems}~\cite{Maurer02}. A \emph{system} is an abstract device taking inputs and
 giving outputs at one or more \emph{interfaces} and is characterized by the 
probability distributions of the outputs given the inputs. 
The closeness of two systems $\mathcal{S}_0$ and $\mathcal{S}_1$ can be measured by introducing a so-called \emph{distinguisher}. A distinguisher $\mathcal{D}$ is itself a system which cna interact with another system 
and output a bit, $B$.  Assume the distinguisher is connected at random either to system $\mathcal{S}_0$ or to $\mathcal{S}_1$; after interacting with the system, 
the distinguisher is supposed to guess which of the two systems it is connected to. 
 The \emph{distinguishing advantage between system $\mathcal{S}_0$ and $\mathcal{S}_1$} is then defined in terms 
 of the probability of winning this game.  
\begin{definition}
The \emph{distinguishing advantage between two systems $\mathcal{S}_0$ and $\mathcal{S}_1$ }is 
\begin{eqnarray}
 \nonumber \delta(\mathcal{S}_0, \mathcal{S}_1)&=& \max_{\mathcal{D}}[P(B=1|\mathcal{S}=\mathcal{S}_1)-P(B=1|\mathcal{S}=\mathcal{S}_0)]\ ,
\end{eqnarray}
where the maximum ranges over all distinguishers $\mathcal{D}$ connected to a system $\mathcal{S}$ and where $B$ denotes 
the output of the distinguisher. 
Two systems $\mathcal{S}_0$ and $\mathcal{S}_1$ are called \emph{$\epsilon$-indistinguishable} if $\delta(\mathcal{S}_0, \mathcal{S}_1)\leq \epsilon$.
\end{definition}
The probability of any event $\mathcal{E}$, defined in a scenario involving a system $\mathcal{S}_0$ cannot 
differ by more than this quantity from the probability of a corresponding event in a scenario where $\mathcal{S}_0$ 
has been replaced by $\mathcal{S}_1$. 
The reason is that otherwise this event could be used to distinguish the two systems. 
\begin{lemma}\label{lemma:event}
Let $\mathcal{S}_0$ and $\mathcal{S}_1$ be $\epsilon$-indistinguishable systems. 
Denote by $P(\mathcal{E}|\mathcal{S}_0)$ the probability of an event $\mathcal{E}$, defined by any of the input and output variables of the system $\mathcal{S}_0$. Then
$
  P(\mathcal{E}|\mathcal{S}_0)\leq  P(\mathcal{E}|\mathcal{S}_1)+ \epsilon
$.
\end{lemma}

The distinguishing advantage is a \emph{pseudo-metric}, in particular, it fulfils the triangle inequality 
\begin{eqnarray}
\label{eq:triangle} \delta(\mathcal{S}_0,\mathcal{S}_1)+\delta(\mathcal{S}_1,\mathcal{S}_2)&\geq & \delta(\mathcal{S}_0,\mathcal{S}_2)\ .
\end{eqnarray}

\subsection{Modelling key agreement}

The security of a cryptographic primitive can be measured by its distance from an \emph{ideal} system which is secure by definition. For example in the case of 
key distribution, the ideal system is a \emph{perfect key generation system} which outputs a uniform and random key (bit string) $S$ to both legitimate parties (usually called Alice and Bob) 
but does not leak any information about $S$ to the adversary (called Eve). 
This key is secure by construction. 
A real key generation system is called \emph{secure} if it is indistinguishable from this ideal one. 

\begin{definition}
A \emph{perfect $\ell$-bit key generation system} is a system which outputs two equal uniform random variables $S_A$ and $S_B$ with range $\mathcal{S}$ of size $|\mathcal{S}|=2^{\ell}$ at two designated interfaces (i.e., 
$P_{S_AS_B}(s_A,s_B)=1/|\mathcal{S}|$ if $s_A=s_B$ and $0$ otherwise) and for which all other interfaces are uncorrelated with $S_A$ and $S_B$. 
\end{definition}

\begin{definition}\label{def:secure}
A key generation system is \emph{$\epsilon$-secure} if the system is $\epsilon$-indistinguish\-able from a perfect key generation system.
\end{definition}
As a consequence of Lemma~\ref{lemma:event}, the resulting security is \emph{composable}~\cite{pw},\,  \cite{bpw},\, \cite{canetti}. That is, no matter 
in what application the key is used, it is as useful as a perfect key, except with a small probability $\epsilon$. 

The real key generation system we are interested in is one obtained by running a protocol $(\pi,\pi^{\prime})$ using as 
underlying resource a public authenticated channel and a pre-distributed quantum state (see Figure~\ref{fig:our_system_physical}). 
More precisely, Alice and Bob both execute locally a program, $\pi$ and $\pi^{\prime}$, to generate keys, $S_A$ and $S_B$, respectively. 
Furthermore, we model the adversary, Eve, as a program that has access to additional interfaces of these 
resources. The interface to the public channel provides her with the entire public communication, $Q$. Furthermore, she can choose an arbitrary measurement, $W$, to be applied to the pre-distributed quantum state, 
resulting in an outcome $Z$. 

Note that, in a realistic scenario, an adversary may access the channel interactively, make measurements and, depending on the outcomes, decide on 
further actions. This is, however, captured by our model, as the single input $W$ may be interpreted 
as the encoding of an entire 
strategy that specifies how a real system would be accessed. 
More precisely, Eve  obtains all the 
communication exchanged over the public channel $Q$, can 
then choose a measurement $W$ (which can depend on 
$Q$) and finally obtains an outcome $Z$. 
\begin{figure}[ht!]
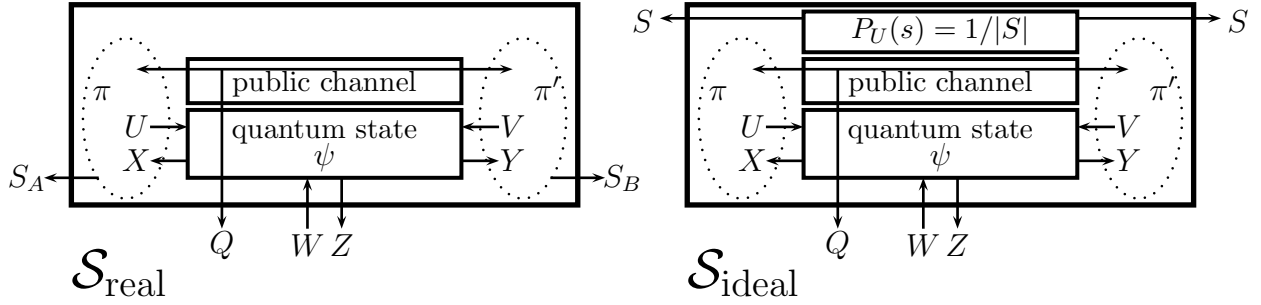

\begin{center}
\pspicture*[](-8.5,-0.5)(8.5,3.5)
\psset{unit=0.9cm}
\rput[c]{0}(-4.5,0){
\rput[l]{0}(-3.7,-0.15){\huge{$\mathcal{S}_{\mathrm{real}}$}}
\psline[linewidth=1pt]{<->}(-2.75,2.85)(2.75,2.85)
\rput[b]{0}(0,2.45){public channel}
\pspolygon[linewidth=1.5pt](-2,3)(2,3)(2,2.35)(-2,2.35)
\pspolygon[linewidth=1.5pt](-2,1.25)(2,1.25)(2,2.25)(-2,2.25)
\psline[linewidth=1pt]{<-}(-2.55,1.5)(-2,1.5)
\rput[b]{0}(-2.75,1.35){\large{$X$}}
\psline[linewidth=1pt]{->}(-2.55,2)(-2,2)
\rput[b]{0}(-2.75,1.85){\large{$U$}}
\psline[linewidth=1pt]{<-}(2.55,1.5)(2,1.5)
\rput[b]{0}(2.75,1.35){\large{$Y$}}
\psline[linewidth=1pt]{->}(2.55,2)(2,2)
\rput[b]{0}(2.75,1.85){\large{$V$}}
\psline[linewidth=1pt]{->}(-0.25,0.5)(-0.25,1.25)
\rput[b]{0}(-0.25,0.1){\large{$W$}}
\psline[linewidth=1pt]{<-}(0.25,0.5)(0.25,1.25)
\rput[b]{0}(0.25,0.1){\large{$Z$}}
\rput[b]{0}(0,1.75){quantum state}
\rput[b]{0}(0,1.35){\large{$\psi$}}
\psline[linewidth=1pt]{<-}(-1.5,0.5)(-1.5,2.85)
\rput[b]{0}(-1.5,0.05){\large{$Q$}}
\psellipse[linewidth=1pt,linestyle=dotted](-2.9,2.125)(0.65,1.2)
\psellipse[linewidth=1pt,linestyle=dotted](2.9,2.125)(0.65,1.2)
\pspolygon[linewidth=2pt](-3.7,0.85)(3.7,0.85)(3.7,3.8)(-3.7,3.8) 
\rput[b]{0}(-3.25,2.4){\large{$\pi$}}
\rput[b]{0}(3.25,2.4){\large{$\pi^{\prime}$}}
\psline[linewidth=1pt]{<-}(-4.1,1.25)(-3.3,1.25)
\rput[b]{0}(-4.35,1.05){\large{$S_A$}}
\psline[linewidth=1pt]{<-}(4.1,1.25)(3.3,1.25)
\rput[b]{0}(4.35,1.05){\large{$S_B$}}
}
\rput[c]{0}(4.5,0){
\rput[l]{0}(-3.7,-0.15){\huge{$\mathcal{S}_{\mathrm{ideal}}$}}
\psline[linewidth=1pt]{<->}(-2.75,2.85)(2.75,2.85)
\rput[b]{0}(0,2.45){public channel}
\pspolygon[linewidth=1.5pt](-2,3)(2,3)(2,2.35)(-2,2.35)
\pspolygon[linewidth=1.5pt](-2,1.25)(2,1.25)(2,2.25)(-2,2.25)
\psline[linewidth=1pt]{<-}(-2.55,1.5)(-2,1.5)
\rput[b]{0}(-2.75,1.35){\large{$X$}}
\psline[linewidth=1pt]{->}(-2.55,2)(-2,2)
\rput[b]{0}(-2.75,1.85){\large{$U$}}
\psline[linewidth=1pt]{<-}(2.55,1.5)(2,1.5)
\rput[b]{0}(2.75,1.35){\large{$Y$}}
\psline[linewidth=1pt]{->}(2.55,2)(2,2)
\rput[b]{0}(2.75,1.85){\large{$V$}}
\psline[linewidth=1pt]{->}(-0.25,0.5)(-0.25,1.25)
\rput[b]{0}(-0.25,0.1){\large{$W$}}
\psline[linewidth=1pt]{<-}(0.25,0.5)(0.25,1.25)
\rput[b]{0}(0.25,0.1){\large{$Z$}}
\rput[b]{0}(0,1.75){quantum state}
\rput[b]{0}(0,1.35){\large{$\psi$}}
\psline[linewidth=1pt]{<-}(-1.5,0.5)(-1.5,2.85)
\rput[b]{0}(-1.5,0.05){\large{$Q$}}
\psellipse[linewidth=1pt,linestyle=dotted](-2.9,2.125)(0.65,1.2)
\psellipse[linewidth=1pt,linestyle=dotted](2.9,2.125)(0.65,1.2)
\pspolygon[linewidth=2pt](-3.7,0.85)(3.7,0.85)(3.7,3.8)(-3.7,3.8) 
\rput[b]{0}(-3.25,2.4){\large{$\pi$}}
\rput[b]{0}(3.25,2.4){\large{$\pi^{\prime}$}}
\pspolygon[linewidth=1.5pt](-2,3.1)(2,3.1)(2,3.7)(-2,3.7)
\rput[b]{0}(0,3.2){$P_U(s)=1/|S|$}
\psline[linewidth=1pt]{<-}(-4.1,3.6)(-2,3.6)
\rput[b]{0}(-4.35,3.4){\large{$S$}}
\psline[linewidth=1pt]{<-}(4.1,3.6)(2,3.6)
\rput[b]{0}(4.35,3.4){\large{$S$}}
}
\endpspicture
\end{center}
\caption{\label{fig:our_system_physical} Our \emph{real} key distribution system $\mathcal{S}_{\mathrm{real}}$ (left). Alice and Bob share a public authentic channel and a quantum state. 
When they apply a protocol $\pi$ to obtain a key, all this can together be modelled as a system. In our \emph{ideal} system $\mathcal{S}_{\mathrm{ideal}}$ (right), instead 
of outputting the key generated by the protocol $(\pi,\pi^{\prime})$, the system outputs a uniform random string $S$ to both Alice and Bob. We will also sometimes use an \emph{intermediate} 
system $\mathcal{S}_{\mathrm{int}}$ which is the same as the real system but with $S_B$ replaced by $S_A$.}
\end{figure}

In order to derive bounds on the parameter $\epsilon$ in Definition~\ref{def:secure}, we split this parameter in two 
parts, where one corresponds to the \emph{correctness} (the probability that Alice's and Bob's key are different, i.e., 
$P(S_A\neq S_B)$) and the other one corresponds to the \emph{secrecy}. The latter is quantified by the 
\emph{distance from uniform} of the key $S_A$ given the information accessible to Eve, i.e., $Z(W_q)$ and $Q$ (we write 
$Z(W_q)$ because the eavesdropper 
can choose the input adaptively and the choice of input changes the output distribution). 
\begin{definition}\label{def:distance_from_uniform}
For a given system as depicted in Figure~\ref{fig:our_system_physical}, the \emph{distance from uniform of $S_A$ given $Z(W_q)$ and $Q$} is 
\begin{eqnarray}
\nonumber
 d(S_A|Z(W_q),Q)
&=&1/2\cdot
\sum_{s,q}  \max_{w} \sum_{z} P_{Z,Q|W=w}(z,q)\cdot|P_{S_A|Z=z,Q=q,W=w}(s)-P_U(s)|
\ ,
\end{eqnarray}
where $P_U$ is the uniform distribution over $|S|$, i.e., $P_U(s)=1/|S|$. 
\end{definition}

The distance from uniform can be seen as the distinguishing advantage between the real system and an intermediate system, $\mathcal{S}_{\mathrm{int}}$, which is equal to our real system, 
but which outputs $S_A$ on both sides (i.e., $S_B$ is replaced by $S_A$). 

The following lemma is a direct consequence of the definitions of the systems in Figure~\ref{fig:our_system_physical} and the distinguishing advantage.
\begin{lemma}
\label{def:dist_advantage} Consider the intermediate system $\mathcal{S}_{\mathrm{int}}$ and the ideal system as defined above (see Figure~\ref{fig:our_system_physical}). 
Then
\begin{eqnarray}
\nonumber 
\delta(\mathcal{S}_{\mathrm{int}},\mathcal{S}_{\mathrm{ideal}})&=&  d(S_A|Z(W_q),Q)\ .
\end{eqnarray}
\end{lemma}
The \emph{correctness} of the protocol, i.e., the probability that Alice's and Bob's key are equal, is determined by the distinguishing advantage 
from the intermediate system to the real system, more precisely, the probability that the real system outputs different values on the two sides. This is again a 
direct consequence of the definitions. 
\begin{lemma}
Consider the intermediate system $\mathcal{S}_{\mathrm{int}}$ and the real system  $\mathcal{S}_{\mathrm{real}}$ as defined above. 
Then
\begin{eqnarray}
\nonumber 
\delta(\mathcal{S}_{\mathrm{real}},\mathcal{S}_{\mathrm{int}})&=&
\sum_{s_A\neq s_B} P_{S_AS_B}(s_A,s_B)\ .
\end{eqnarray}
\end{lemma}

Finally, by the triangle inequality (\ref{eq:triangle}) on the distinguishing advantage of systems, we obtain the following lemma relating the security of our protocol to the secrecy (measured in terms of the distance from uniform) and the correctness.
\begin{lemma}
The key generation system depicted in Figure~\ref{fig:our_system_physical} is $\epsilon$-secure if
\begin{eqnarray}
\nonumber 
\epsilon &\leq &  d(S_A|Z(W_q),Q) + \sum_{s_A\neq s_B} P_{S_AS_B}(s_A,s_B)\ .
\end{eqnarray}
\end{lemma}

\subsection{Quantum systems}

The fact that in the above scenario (see Figure~\ref{fig:our_system_physical}), the random variables $U,V,W,X,Y,Z$ correspond to the choice of
 measurements on a quantum state and their respective outcomes imposes a limitation 
on their possible distribution and, with this, on the eavesdropper's attacks. 
Consider the scenario where Alice, Bob and Eve share a tripartite quantum state. They can each measure their 
part of the system and obtain a measurement outcome. 
We can, of course, also consider the system Alice and Bob share tracing out Eve and this still corresponds to 
a quantum state (the reduced state). 
In accordance with the non-signalling principle, the 
marginal state Alice and Bob share is independent of what Eve does with her part of the state (in particular, 
independent of her measurement). And we can even consider 
the state Alice and Bob share conditioned an a certain measurement outcome of Eve: Alice and Bob still share a 
quantum state in this case. 

\begin{definition}\label{def:qbehavior}
 An $n$-party system $P_{\bof{X}|\bof{U}}$, where $\bof{X}=(X_1\ldots X_n)$, is called \emph{quantum} if there exists a pure state 
$\ket{\psi}\in \mathcal{H}=\bigotimes_i \mathcal{H}_i$ and a set of measurement operators $\{E_{u_i}^{x_i}\}$ on $\mathcal{H}_i$ such that
\begin{eqnarray}
\nonumber  P_{\bof{X}|\bof{U}}(\bof{x},\bof{u})&=& \brakket{\psi}{\bigotimes_i E_{u_i}^{x_i}}{\psi}\ ,
\end{eqnarray}
where the measurement operators satisfy the following conditions 
\begin{enumerate}
 \item Hermitian, i.e., ${E_{u_i}^{x_i}}^{\dagger}=E_{u_i}^{x_i}$ 
 for all $x_i,u_i$, 
 \item orthogonal projectors, i.e., $E_{u_i}^{x_i} E_{u_i}^{x^{\prime}_i}=E_{u_i}^{x_i}\delta_{x_ix^{\prime}_i}$, 
 \item and sum up to the identity, i.e., $\sum_{x_i}{E_{u_i}^{x_i}}=\mathds{1}_{\mathcal{H}_i}$ 
  for all $u_i$.
\end{enumerate}
\end{definition}
Note that the requirement that the operators correspond to projectors and the state to a pure state is not a restriction, since any POVM on a mixed state is 
equivalent to a projective measurement on a larger pure state (see, e.g.,~\cite{nielsenchuang} for a proof). 

For any $(n+1)$-party quantum system, the marginal and conditional systems are also quantum systems. 
\begin{lemma}\label{lemma:qmarginalconditional}
Consider an $(n+1)$-party quantum system $P_{\bof{X}Z|\bof{U}W}$. Then the marginal system 
\begin{eqnarray}
\nonumber P_{\bof{X}|\bof{U}}(\bof{x},\bof{u})&:= &
\sum_z P_{\bof{X}Z|\bof{U}W}(\bof{x},z,\bof{u},w)
\end{eqnarray}
and the conditional system 
\begin{eqnarray}
\nonumber P_{\bof{X}|\bof{U},W=w,Z=z}(\bof{x},\bof{u})&:=&
\frac{1}{P_{Z|W=w}(z)}P_{\bof{X}Z|\bof{U}W}(\bof{x},z,\bof{u},w) 
\end{eqnarray}
 are $n$-party quantum systems. 
\end{lemma}
This follows, of course, directly from the properties of quantum systems. However, as an illustration, we give a direct proof in our framework.
\begin{proof} 
Let $\ket{\psi}$ be the state and $\{E_{u_i}^{x_i}\}$  the measurement operators associated with the original system. 
For the marginal system, take the same state $\ket{\psi}$ and the measurement operators $\{E_{u_i}^{x_i}\}$ for all $i<n$. The measurement operator associated 
with the $n$\textsuperscript{th} party are $\{E_{u_n}^{x_n}\otimes \mathds{1}_{\mathcal{H}_{n+1}}\}$.  They fulfil the requirements because they are part of the 
requirements of the operators of the $(n+1)$-party quantum system. 
For the conditional system take the state 
$\frac{1}{\sqrt{\brakket{\psi}{\mathds{1}_{\mathcal{H}_{1\ldots n}} \otimes E_w^z}{\psi}}} \mathds{1}_{\mathcal{H}_{1\ldots n}} \otimes E_w^z \ket{\psi}$, 
where $\mathds{1}_{\mathcal{H}_{1\ldots n}}=\bigotimes_{i=1}^{n}\mathds{1}_{\mathcal{H}_i}$
 and the measurement operators $\{E_{u_i}^{x_i}\}$.  
\end{proof}

Lemma~\ref{lemma:qmarginalconditional} directly implies that any measurement of Eve on her part of the 
system induces a convex decomposition of Alice's and Bob's system into several conditional quantum systems. 
\begin{remark}\label{remark:qpartition}
Every input to an $(n+1)$-party quantum system,  $P_{\bof{X}Z|\bof{U}W}$, corresponds to a decomposition of the marginal $n$-party system 
 $P_{\bof{X}|\bof{U}}$ such that 
 \begin{equation}
 \nonumber 
 P_{\bof{X}|\bof{U}}=\sum_z p^z\cdot P^z_{\bof{X}|\bof{U}}\ ,
 \end{equation}
where $p^z:= P_{Z|W=w}(z)$ is a probability and $P^z_{\bof{X}|\bof{U}}:=P_{\bof{X}|\bof{U},W=w,Z=z}$ is a quantum system. 
\end{remark}

\section{Bounding the Set of Quantum Systems by Semi-Definite Programming}\label{sec:condition}

In~\cite{npa07}, Navascu\'es, Pironio and Ac\'in give a criterion in terms of a semi-definite program  (see, e.g.~\cite{bv,btn} for an 
introduction to semi-definite programming) which any quantum system must fulfil (see also~\cite{npa,dltw08}). 
The idea is that if a system is quantum, then it is possible to associate a matrix $\Gamma$ with it which needs to be positive semi-definite. We will use the notation $\Gamma\succeq 0$ to denote positive semi-definite matrices.  $\Gamma$ can be seen as the matrix defined as follows.

\begin{definition}
 A \emph{sequence of length $k$} of a set 
$\{E_{u_i}^{x_i}:x_i\in \mathcal{X}_i, u_i\in \mathcal{U}_i,i\in 1,\ldots,n\}$ is 
 a product of $k$ operators of this set. 
  The sequence of length $0$ is defined as the identity operator. 
\end{definition}

\begin{definition}\label{def:gamma}
The matrix \emph{$\Gamma^k$} is defined as 
\begin{eqnarray}
\nonumber \Gamma^k_{ij}:= \brakket{\Psi}{O_i^\dagger O_j}{\Psi}\ ,
\end{eqnarray}
where $O_i=E_{u_m}^{x_m}\cdot E_{u^{\prime}_n}^{x^{\prime}_n}\cdots $ is a sequence of length at most $k$ of the measurement operators $\{E_{u_i}^{x_i}\}$. 
\end{definition}
In the above notation we consider the measurement operators as operators on the whole Hilbert space $\mathcal{H}$. These operators must, of course, fulfil the conditions of 
Definition~\ref{def:qbehavior} (i.e., they must be Hermitian orthogonal projectors and sum up to the identity for each input) and they must commute. Note that in finite 
dimensions, commutativity is equivalent to the tensor product structure as in Definition~\ref{def:qbehavior} (see, e.g.,~\cite{wehnerphd} for an explicit proof of this). 

The requirements the measurement operators fulfil (Definition~\ref{def:qbehavior}) translate into requirements on the entries of the matrix $\Gamma^k$. For example, certain entries must be equal to others 
or the sum of some must be equal to the sum of others. 

In order to decide whether a certain system is quantum, we can ask the question whether such a matrix 
$\Gamma^k$ exists; because if it is, it must be possible to associate a matrix $\Gamma$ (as in Definition~\ref{def:gamma}) with it, which is consistent 
with the probabilities describing the system and fulfil the above requirements. The problem of finding a consistent matrix $\Gamma^k$ is a semi-definite programming problem. 
\begin{theorem}[Navascu\'es, Pironio, Ac\'in~\cite{npa07}]
For every quantum system $P_{\bof{X}|\bof{U}}$ and $k\in \mathbb{N}$ there exists a symmetric matrix $\Gamma^k$ with 
$\Gamma^k_{ij}=\brakket{\Psi}{O_i^\dagger O_j}{\Psi}$ and where $O_i=E_{u_m}^{x_m}\cdot E_{u^{\prime}_n}^{x^{\prime}_n}\cdots $ is a sequence of length $k$. 
Furthermore, 
\begin{eqnarray}
\nonumber A_{\mathrm{qb}}\cdot \Gamma^k&=&0\ \text{and}\\
\nonumber \Gamma^k &\succeq & 0\ ,
\end{eqnarray}
where $A_{\mathrm{qb}}$ is defined by the conditions 
\begin{itemize}
 \item orthogonal projectors: $\brakket{\Psi}{O E_{u_i}^{x_i} E_{u_i}^{x^{\prime}_i}O^{\prime}}{\Psi}-\brakket{\Psi}{OE_{u_i}^{x_i}\delta_{x_ix^{\prime}_i} O^{\prime}}{\Psi}=0$  
 \item completeness: $\sum_{x_i}\brakket{\Psi}{OE_{u_i}^{x_i} O^{\prime}}{\Psi}-\brakket{\Psi}{O O^{\prime}}{\Psi}=0$ for all ${u}_i$ 
 \item commutativity: $\brakket{\Psi}{O {E_{u_i}^{x_i}} {E_{u_j}^{x_j}} O^{\prime}}{\Psi}=\brakket{\Psi}{O{E_{u_j}^{x_j}} {E_{u_i}^{x_i}}O^{\prime}}{\Psi}$ for $i\neq j$,
\end{itemize}
where  $O$ and $O^{\prime}$ stand for arbitrary operator sequences of the set $\{E_{u_i}^{x_i}\}$.
 
$\Gamma^k$ is called \emph{quantum certificate of order $k$} associated with the system $P_{\bof{X}|\bof{U}}$. 
\end{theorem}
\begin{proof}
Orthogonality, completeness and Hermiticity follow directly from Definition~\ref{def:qbehavior}. Let us see that the matrix is positive semi-definite. 
 For all $v\in \mathbb{C}^m$
\begin{eqnarray}
 \nonumber v^T \Gamma^k v=\sum_{ij}v_i^T \Gamma^k_{ij} v_i=\sum_{ij}v_i^* \brakket{\Psi}{O_i^\dagger O_j}{\Psi} v_j =\brakket{\Psi}{V^\dagger V}{\Psi}\geq 0
\end{eqnarray}
where $V:=\sum_iv_iO_i$. Finally, the matrix can be taken to be real, because for any complex $\Gamma^k$, the matrix $(\Gamma^k+{\Gamma^k}^*)/2$ is real 
and fulfils the conditions. 
\end{proof}
We do not require this matrix to be normalized. Note that the matrix $\Gamma^k$ contains, in particular, the (potentially not normalized) probabilities 
$P_{\bof{X}|\bof{U}}(\bof{x},\bof{u})$ associated with an $n$-party quantum system, for $n\leq 2k$. 

In~\cite{npa,dltw08}, it is shown that if for all $k\rightarrow \infty$ a certificate of order $k$ can be associated with a certain system $P_{\bof{X}|\bof{U}}$, 
then this system is indeed quantum. More precisely, it corresponds to a quantum system where operators associated with different parties commute, but do not necessarily 
have a tensor product structure. For any finite dimensional system however, commutativity implies a tensor product structure. See, e.g.,~\cite{dltw08} for an explicit proof of this.

\section{Min-Entropy Bound for Single Systems}\label{sec:qsingle}

It will be our goal to show the security of a key-distribution protocol of the form as given in Figure~\ref{fig:our_system_physical}. The crucial part hereby is 
to bound the probability that an eavesdropper interacting with her part of the quantum state can correctly guess the value of Alice's raw key $\bof{X}$, since this corresponds to the min-entropy, by the following theorem.
\begin{theorem}[K\"onig, Renner, Schaffner~\cite{krs}]\label{th:krs}
Let $\rho_{XE}$ be classical on $\mathcal{H}_X$. Then 
\begin{eqnarray}
\nonumber H_{\mathrm{min}}(X|E)_{\rho} &=&-\log_2 P_{\mathrm{guess}}(X|E)_{\rho}\ ,
\end{eqnarray}
where $p_{\mathrm{guess}}(X|E)_{\rho}$ is the maximal probability of decoding $X$ from $E$ with a POVM $\{E_E^x\}_x$ on $\mathcal{H}_E$, i.e., 
\begin{eqnarray}
\nonumber P_{\mathrm{guess}}(X|E)_{\rho}:= \max_{\{E_E^x\}_x}\sum_x p_x \tr(E_E^x\rho_B^x)\ .
\end{eqnarray}
\end{theorem}
 This implies that, in order to bound the min-entropy, we can equivalently bound the guessing probability. Once the min-entropy is bounded, a secure key can be obtained using standard techniques, such as information reconciliation~\cite{brassardsalvail} and privacy amplification, which work even if the adversary holds quantum information~\cite{rennerkoenig,rennerphd}.

In this section, we will see how it is possible to determine the security of a single system (corresponding to a single measurement of Alice and Bob) by a semi-definite program. In Section~\ref{sec:several}, we will see that 
the security of many systems, and therefore of the key distribution scheme, directly relates to the security of the single system.

In the following, we will often consider a $(2n+1)$-party quantum system $P_{\bof{X}\bof{Y}Z|\bof{U}\bof{V}W}$ (as well as its marginals) where 
$\bof{U}=(U_1\ldots U_n)$ and $\bof{X}=(X_1\ldots X_n)$ are Alice's inputs and outputs, 
$\bof{V}=(V_1\ldots V_n)$ and $\bof{Y}=(Y_1\ldots Y_n)$ are Bob's inputs and outputs, and 
$W$ and $Z$  Eve's input and output. 
The fact that Eve only has a single input and output reflects the fact that Eve may perform 
a \emph{joint attack}, which means that she would not necessarily measure her subsystems individually.

\subsection{A bound on the min-entropy}\label{subsec:qguess}

We will, in the following, study the scenario where Eve can choose an input $W$,  depending on some additional information $Q$, and then obtains an output $Z$ 
(depending on $W$). She should then try to guess a value $f(\bof{X})$ of range $\mathcal{F}$. In particular, this function $f$ can, of course, be the identity function on the outputs on Alice's side.  

\begin{definition}
The \emph{guessing probability of $f(\bof{X})$ given $Z(W)$} is 
\begin{eqnarray}
\nonumber P_{\mathrm{guess}}(f(\bof{X})|Z(W),Q)&=& \sum_q \max_{w} \sum_z P_{ZQ|W=w}(z,q)
\cdot
 \max_{f(x)} P_{f(X)|Z=z,Q=q,W=w}(f(x))\ ,
\end{eqnarray}
where the maximization is over all $w$ such that $P_{\bof{X}Z|\bof{U}W}$ is a quantum system. The \emph{min-entropy of $f(\bof{X})$ given $Z(W)$} is 
\begin{eqnarray}
\nonumber H_{\mathrm{min}}(f(\bof{X})|Z(W),Q))=-\log_2 P_{\mathrm{guess}}(f(\bof{X})|Z(W),Q))\ .
\end{eqnarray}
\end{definition}

Remark~\ref{remark:qpartition} gives a bound on the probability that a quantum adversary can guess Alice's outcome by the following maximization problem. (We assume 
that the inputs $\bof{u}$ are public, i.e., $Q=(\bof{U}=\bof{u},F=f)$)
\begin{lemma}\label{lemma:qattackopt}
The value of $P_{\mathrm{guess}}(f(\bof{X})|Z(W),Q)$ where $P_{\bof{X}Z|\bof{U}W}$ is a $(n+1)$-party quantum system and $Q=(\bof{U}=\bof{u})$ is bounded 
by the optimal value of the following optimization problem
\begin{eqnarray}
\nonumber \max :&& \sum_{z=1}^{|\mathcal{F}|} p^z \sum_{\bof{x}:f(\bof{x})=z} P^z_{\bof{X}|\bof{U}}(\bof{x},\bof{u})\\
\nonumber \st : && P_{\bof{X}|\bof{U}}= \sum_{z=1}^{|\mathcal{F}|}  p^z\cdot P^z_{\bof{X}|\bof{U}} \\
\nonumber &&P^z_{\bof{X}|\bof{U}}\ n\text{-party quantum system, for all }z\ .
\end{eqnarray}
\end{lemma}
\begin{proof}
The first condition follows by the definition of the marginal system and the second by the fact that for any $(n+1)$-party 
quantum system the conditional systems are $n$-party quantum systems (see Lemma~\ref{lemma:qmarginalconditional}). The objective function is the definition of guessing 
probability. 
It is sufficient to consider the case $|\mathcal{Z}|=|\mathcal{F}|$ because any system where $Z$ has larger range can be made into a system reaching the same guessing 
probability by combining the system where the same value $f(\bof{X})$ has maximal probability. By the convexity of quantum systems, this is still a quantum system. 
\end{proof}

The criterion discussed in Section~\ref{sec:condition} allows to replace the condition that $P_{\bof{XY}|\bof{UV}}^z$ is a quantum behaviour by the condition that a 
certain matrix is positive semi-definite. We can now bound Eve's guessing probability by a semi-definite program. 
A similar bound has been obtained in~\cite{random} 
in the context of device-independent randomness expansion.
\begin{lemma}\label{lemma:guessissdp}
The maximum guessing probability of $f(\bof{X})$ given $Z(W)$ and $Q:=(\bof{U}=\bof{u},F=f)$ is bounded by\footnote{In the following, we sometimes write 
matrices as vectors by writing each column `on top of each other'. When we write that a vector needs to be positive semi-definite, we mean that the matrix obtained 
by the inverse transformation must be positive semi-definite.}
\begin{eqnarray}
\nonumber P_{\mathrm{guess}}(f(\bof{X})|Z(W),Q)&\leq &  \sum_{z=1}^{|\mathcal{F}|} b_z^T\cdot \Gamma^{z}\ ,
\end{eqnarray}
where  $\sum_{z=1}^{|\mathcal{F}|} b_z^T\cdot \Gamma^{z}$ is the optimal value of the semi-definite program
\begin{eqnarray}
\max :&& \sum_{z=1}^{|\mathcal{F}|} \sum_{\bof{x}:f(\bof{x})=z} \Gamma^{z}(\bof{x},\bof{u})\\
\nonumber \st :&& A_{\mathrm{qb}}\cdot \Gamma^z=0\ \text{ for all } z 
\\
\nonumber && \Gamma^z\succeq 0\\
\nonumber && \sum_z \Gamma^z = \Gamma^k_{\mathrm{marg}}
\end{eqnarray}
where $\Gamma^{z}(\bof{x},\bof{u})$ denotes the entry of the matrix $\Gamma^z$ corresponding to $\brakket{\Psi}{\prod_i E_{u_i}^{x_i} E_w^z}{\Psi}$, i.e., the 
probability $P^z_{\bof{X}|\bof{U}}(\bof{x},\bof{u})$; $b_z$ is a matrix of the same size as $\Gamma^z$ and it has a $1$ at the positions where $\Gamma^k$ has the 
entry $\brakket{\psi}{O_i^\dagger O_i}{\psi}$, where $O_i=\prod_m E_{u_m}^{x_m}$ 
 such that $f(\bof{x})=z$.  $\Gamma^k_{\mathrm{marg}}$ denotes the certificate of order $k$ associated with the marginal system $P_{\bof{X}|\bof{U}}$. 
\end{lemma}
\begin{proof}
This follows from Lemma~\ref{lemma:qattackopt}, the fact that any quantum system $P_{\bof{X}|\bof{U}}^z$ has a quantum certificate of order $k$ and $\sum_z E_w^z=\mathds{1}$. 
\end{proof}

The primal and dual program can be expressed as:
\begin{eqnarray}
\nonumber \text{PRIMAL}\\
\label{eq:primalguess}\max :& \sum_{z=1}^{|\mathcal{F}|}  b_z^T\cdot \Gamma_z\\
\nonumber\\
\nonumber\st :& 
\underbrace{
\left(
\begin{array}{ccc}
A_{\mathrm{qb}} & \cdots  & 0\\
& \ddots & \\
0 & \cdots & A_{\mathrm{qb}}\\
\mathds{1} & \cdots &\mathds{1}
\end{array}
\right)
}_{A}
\cdot 
\left( 
\begin{array}{c}
 \Gamma_1\\
\vdots \\
\Gamma_{|\mathcal{F}|}
\end{array}
\right)
=
\underbrace{
\left(
\begin{array}{c}
0\\
\vdots \\
0 \\
\Gamma_{\mathrm{marg}}^k
\end{array}
\right)
}_{c} \\
\nonumber \\
\nonumber& \Gamma_i \succeq 0\ \text{ for all } i
\end{eqnarray}
\begin{eqnarray}
\nonumber \text{DUAL}\\
\label{eq:dualguess} \min :& {\Gamma^k_{\mathrm{marg}}}^T\cdot  \lambda_{{|\mathcal{F}|}+1}\\
\nonumber \\
 \st :& 
 \underbrace{
\left(
\begin{array}{cccc}
A_{\mathrm{qb}}^T & \cdots  & 0 & \mathds{1} \\
& \ddots &&\\
0 & \cdots & A_{\mathrm{qb}}^T & \mathds{1}\\
\end{array}
\right)
}_{A^T}
\cdot 
\left(
\begin{array}{c}
\lambda_1 \\
\vdots \\
\lambda_{|\mathcal{F}|}\\
\lambda_{{|\mathcal{F}|}+1}
\end{array}
\right)
\succeq 
\underbrace{
\left(
\begin{array}{c}
b_1 \\
\vdots \\
b_{|\mathcal{F}|}
\end{array}
\right)
}_{b}
\nonumber \\
\nonumber \\
\nonumber &\lambda_i\ \text{unrestricted}
\end{eqnarray}
We note that any dual feasible solution gives an upper bound on the guessing probability (linear) in terms of the matrix $\Gamma_{\mathrm{marg}}$ associated with the marginal 
system of Alice and Bob. Furthermore, the dual feasible region is  \emph{independent} of Alice's and Bob's marginal system, it only depends on the number of inputs and outputs 
and the step in the semi-definite hierarchy considered. 

However, the matrix $\Gamma^k_{\mathrm{marg}}$ contains entries which do not correspond to observable probabilities and are only known if the state and measurement operators are 
known. It will be the goal of the next section to express the guessing probability in terms of observable quantities.

\subsection{A min-entropy bound in terms of observable probabilities}\label{subsec:observableprob}

Certain entries of the matrix $\Gamma^k_{\mathrm{marg}}$ do not correspond to observable probabilities and it is, therefore, impossible to know their value 
by testing the system. In this section, we will modify the above optimization problem in such a way as to get a solution only in terms of observable probabilities. 
More precisely, we will modify the optimization problem to take the `worst' possible quantum certificate  consistent with observed probabilities. This leads to the 
following, modified, semi-definite program. The matrix $A_{\mathrm{IJ}}$ is defined such that multiplied with a quantum certificate the observable probabilities are 
obtained, i.e., $A_{\mathrm{IJ}}\cdot \Gamma^k=P_{\bof{X}|\bof{U}}$ (where $P_{\bof{X}|\bof{U}}$ denotes here the vector containing the values $P_{\bof{X}|\bof{U}}(\bof{x},\bof{u})$ 
for all $\bof{x},\bof{u}$).
\begin{eqnarray}
\nonumber \text{PRIMAL}\\
\label{eq:primalguessprob}
\max :& \sum_{z=1}^{|\mathcal{F}|} b_z^T\cdot \Gamma_z\\
\nonumber \\
\nonumber \st :& \left(
\begin{array}{cccc}
A_{\mathrm{qb}} & \cdots  & 0 & \phantom{-}0\\
 & \ddots & & \phantom{-}0\\
0 & \cdots & A_{\mathrm{qb}} & \phantom{-}0\\
\mathds{1} & \cdots & \mathds{1} & -\mathds{1}\\
0 & \cdots &0 & A_{\mathrm{IJ}}\\
\end{array}
\right)
\cdot 
\left( 
\begin{array}{c}
 \Gamma_1\\
\vdots \\
\Gamma_{|\mathcal{F}|}\\
\Gamma^k_{\mathrm{marg}}
\end{array}
\right)
=
\left(
\begin{array}{c}
0\\
\vdots \\
0\\
0 \\
P_{\bof{X}|\bof{U}}
\end{array}
\right)\\
\nonumber \\
\nonumber & \Gamma_i \succeq 0,\ \Gamma^k_{\mathrm{marg}}\ \text{unrestricted}
\end{eqnarray} 
\begin{eqnarray}
\nonumber \text{DUAL}\\
\label{eq:dualguessprob}
\min : & P_{\bof{X}|\bof{U}}^T \cdot \lambda_{|\mathcal{F}|+2}\\
\nonumber \\
 \st :& 
\left(
\begin{array}{ccccc}
A_{\mathrm{qb}}^T & \cdots  & 0 & \phantom{-}\mathds{1} & 0 \\
& \ddots & 0 &\\
0 & \cdots & A_{\mathrm{qb}}^T & \phantom{-}\mathds{1} & 0\\
0 & \cdots & 0 & -\mathds{1} & A_{\mathrm{IJ}}^T\\
\end{array}
\right)
\cdot 
\left(
\begin{array}{c}
\lambda_1 \\
\vdots \\
\lambda_{|\mathcal{F}|}\\
\lambda_{|\mathcal{F}|+1}\\
\lambda_{|\mathcal{F}|+2}
\end{array}
\right)
\begin{array}{c}
\\
 \succeq \\
\\
=
\end{array}
\left(
\begin{array}{c}
b_1 \\
\vdots \\
b_{|\mathcal{F}|} \\
0
\end{array}
\right)
\nonumber  \\
\nonumber  &\lambda_i\ \text{unrestricted}
\end{eqnarray}
Note that we have changed $\Gamma^k_{\mathrm{marg}}$ to be a variable (instead of a constant). Obviously $\Gamma^k_{\mathrm{marg}}\succeq 0$ holds because it is the 
sum of positive semi-definite matrices.

\begin{lemma}\label{lemma:qguessdualprod}
 If $\lambda_1,\ldots,\lambda_{|\mathcal{F}|+2}$ are dual feasible for (\ref{eq:dualguessprob}), then 
$\lambda_1,\ldots,\lambda_{|\mathcal{F}|+1}$ are dual feasible for (\ref{eq:dualguess}) with the same objective value. 
\end{lemma}
\begin{proof}
We use the fact that $A_{\mathrm{IJ}}\cdot \Gamma^k_{\mathrm{marg}}=P_{\bof{X}|\bof{U}}$. 
Since $\lambda_1,\ldots,\lambda_{|\mathcal{F}|+2}$ are dual feasible for (\ref{eq:dualguessprob}), it holds that $A_{\mathrm{IJ}}^T\cdot \lambda_{|\mathcal{F}|+2}
=\lambda_{|\mathcal{F}|+1}$. Therefore, 
\begin{eqnarray}
\nonumber {\Gamma^k_{\mathrm{marg}}}^T\cdot \lambda_{|\mathcal{F}|+1}&=& {\Gamma^k_{\mathrm{marg}}}^T\cdot A_{\mathrm{IJ}}^T\cdot \lambda_{|\mathcal{F}|+2} = 
P_{\bof{X}|\bof{U}}^T\cdot \lambda_{|\mathcal{F}|+2} \ .
\end{eqnarray}
\end{proof}
Lemma~\ref{lemma:qguessdualprod} implies that any dual feasible solution of (\ref{eq:dualguessprob}) gives an upper bound on the guessing probability linear in terms 
of the observable probabilities. In terms of the min-entropy we obtain the following corollary. 
\begin{corollary}
For any dual feasible $\lambda$, 
\begin{equation}
\nonumber H_{\mathrm{min}}(\bof{X}|Z(W))\leq -\log_2 (P_{\bof{X}|\bof{U}}^T\cdot \lambda_{|\mathcal{F}|+2})\ .
\end{equation} 
\end{corollary}

\begin{example}\label{ex:qguesssingle}
Consider a bipartite quantum system with binary inputs and outputs given by the mixture of the system in Figure~\ref{fig:box1} with weight $1-\rho$ and a perfectly 
random bit with weight $\rho$ (this could be achieved by measuring a mixture of a singlet and a fully mixed state, i.e., the state 
$(1-\rho)\cdot \ket{\Psi^-}\bra{\Psi^-}+\rho\cdot \frac{1}{4}\mathds{1}$
 using the measurements $U_0,U_1,V_0,V_1$ given in Figure~\ref{fig:basen}). The guessing probability of the output bit $X$ as function of the parameter $\rho$ is given 
in Figure~\ref{fig:pguess}.\footnote{The data plotted in Figure~\ref{fig:pguess} has been obtained by solving (\ref{eq:primalguessprob}) numerically, using  the programs MATLAB\textsuperscript{\textregistered}, Yalmip and Sedumi~\cite{matlab,sedumi,yalmip}.}

\begin{figure}[h]
\begin{minipage}[b]{0.45\linewidth}
\centering
\psset{unit=0.525cm}
\pspicture*[](-2,-0.1)(8.5,8)
\psline[linewidth=0.5pt]{-}(0,6)(-1,7)
\rput[c]{0}(-0.25,6.75){\normalsize{$X$}}
\rput[c]{0}(-0.75,6.25){\normalsize{$Y$}}
\rput[c]{0}(-0.5,7.5){\Large{$U$}}
\rput[c]{0}(-1.5,6.5){\Large{$V$}}
\rput[c]{0}(2,7.5){\Large{$0$}}
\rput[c]{0}(6,7.5){\Large{$1$}}
\rput[c]{0}(1,6.5){\Large{$0$}}
\rput[c]{0}(3,6.5){\Large{$1$}}
\rput[c]{0}(5,6.5){\Large{$0$}}
\rput[c]{0}(7,6.5){\Large{$1$}}
\rput[c]{0}(-1.5,4.5){\Large{$0$}}
\rput[c]{0}(-1.5,1.5){\Large{$1$}}
\rput[c]{0}(-0.5,5.25){\Large{$0$}}
\rput[c]{0}(-0.5,3.75){\Large{$1$}}
\rput[c]{0}(-0.5,2.25){\Large{$0$}}
\rput[c]{0}(-0.5,0.75){\Large{$1$}}
\psline[linewidth=2pt]{-}(-1,0)(8,0)
\psline[linewidth=2pt]{-}(-1,6)(8,6)
\psline[linewidth=2pt]{-}(-1,3)(8,3)
\psline[linewidth=1pt]{-}(0,1.5)(8,1.5)
\psline[linewidth=1pt]{-}(0,4.5)(8,4.5)
\psline[linewidth=2pt]{-}(0,0)(0,7)
\psline[linewidth=2pt]{-}(8,0)(8,7)
\psline[linewidth=2pt]{-}(4,0)(4,7)
\psline[linewidth=1pt]{-}(2,0)(2,6)
\psline[linewidth=1pt]{-}(6,0)(6,6)
\rput[c]{0}(1,5.25){\Large{$\frac{2+\sqrt{2}}{8}$}}
\rput[c]{0}(3,3.75){\Large{$\frac{2+\sqrt{2}}{8}$}}
\rput[c]{0}(5,5.25){\Large{$\frac{2+\sqrt{2}}{8}$}}
\rput[c]{0}(7,3.75){\Large{$\frac{2+\sqrt{2}}{8}$}}
\rput[c]{0}(1,2.25){\Large{$\frac{2+\sqrt{2}}{8}$}}
\rput[c]{0}(3,0.75){\Large{$\frac{2+\sqrt{2}}{8}$}}
\rput[c]{0}(5,0.75){\Large{$\frac{2+\sqrt{2}}{8}$}}
\rput[c]{0}(7,2.25){\Large{$\frac{2+\sqrt{2}}{8}$}}
\rput[c]{0}(3,5.25){\Large{$\frac{2-\sqrt{2}}{8}$}}
\rput[c]{0}(1,3.75){\Large{$\frac{2-\sqrt{2}}{8}$}}
\rput[c]{0}(7,5.25){\Large{$\frac{2-\sqrt{2}}{8}$}}
\rput[c]{0}(5,3.75){\Large{$\frac{2-\sqrt{2}}{8}$}}
\rput[c]{0}(3,2.25){\Large{$\frac{2-\sqrt{2}}{8}$}}
\rput[c]{0}(1,0.75){\Large{$\frac{2-\sqrt{2}}{8}$}}
\rput[c]{0}(5,2.25){\Large{$\frac{2-\sqrt{2}}{8}$}}
\rput[c]{0}(7,0.75){\Large{$\frac{2-\sqrt{2}}{8}$}}
\endpspicture
\vspace{1cm}
\caption{\label{fig:box1}The probabilities associated with a quantum system obtained by measuring the singlet state using the bases $U_0,U_1,V_0,V_1$ of Figure~\ref{fig:basen}.}
\end{minipage}
\hspace{0.5cm}
\begin{minipage}[b]{0.45\linewidth}
\centering
 \pspicture[](-2,2.75)(7,7.75)
 \psset{xunit=16.66666666cm,yunit=7.5cm}
 \savedata{\mydata}[
 {
{0.0000,	0.5000},
{0.0030,	0.5546},
{0.0060,	0.5769},
{0.0090,	0.5938},
{0.0120,	0.6079},
{0.0150,	0.6202},
{0.0180,	0.6312},
{0.0210,	0.6412},
{0.0240,	0.6505},
{0.0270,	0.6590},
{0.0300,	0.6670},
{0.0330,	0.6746},
{0.0360,	0.6817},
{0.0390,	0.6885},
{0.0420,	0.6949},
{0.0450,	0.7011},
{0.0480,	0.7070},
{0.0510,	0.7126},
{0.0540,	0.7180},
{0.0570,	0.7233},
{0.0600,	0.7285},
{0.0630,	0.7336},
{0.066,	0.7386},
{0.069,	0.7436},
{0.072,	0.7485},
{0.075,	0.7534},
{0.078,	0.7583},
{0.081,	0.7631},
{0.084,	0.7678},
{0.087,	0.7725},
{0.09,	0.7771},
{0.105,	0.7996},
{0.12,	0.8209},
{0.135,	0.8412},
{0.15,	0.8605},
{0.1650,	0.8788},
{0.1800,	0.8962},
{0.1950,	0.9128},
{0.2100,	0.9284},
{0.2250,	0.9433},
{0.2400,	0.9573},
{0.2550,	0.9704},
{0.2700,	0.9828},
{0.2850,	0.9943},
{0.3000,	1.0000}
 }
 ]
 \rput[c](0.15,0.375){{$\rho$}}
  \rput[c]{90}(-0.075,0.75){{$P_{\mathrm{guess}}$}}
   \psaxes[Dx=0.05,Dy=0.1, Oy=0.5,  showorigin=true,tickstyle=bottom,axesstyle=frame](0,0.5)(0.3001,1.0001)
\dataplot[plotstyle=curve,showpoints=false,dotstyle=o]{\mydata}  
 \endpspicture
 \caption{\label{fig:pguess} The bound on the guessing probability of the measurement outcomes of Example~\ref{ex:qguesssingle}.\\
}
\end{minipage}
\end{figure}

\end{example}

\section{Min-Entropy Bound for Multiple Systems}\label{sec:several}

We can now show our main technical result, namely that the above semi-definite program describing the guessing probability has a product form if the measurements 
on different subsystems commute. Roughly, we will show the following: consider a system $P_{\bof{XY}|\bof{UV}}$ associated with a single pair of systems and the matrix 
$\Gamma^k$ associated with the $k$\textsuperscript{th} step of the hierarchy , 
fulfilling $A_{\mathrm{qb}}\cdot \Gamma^k=0$. 
Then with two pairs of systems it is possible to associate a matrix ${\Gamma^{\prime}}^k$ living in the tensor product space of two $\Gamma^k$. Furthermore, this matrix must 
fulfil $(\mathds{1}\otimes A_{\mathrm{qb}}) {\Gamma^{\prime}}^k=0$. 

\subsection{Conditions on several quantum systems}\label{subsec:qseveral}

The goal of this section is to express the constraints that hold for a multi-party quantum system in terms of the 
constraints on its subsystems. 

\begin{definition}
Assume an $(n+m)$-party quantum system. The \emph{reduced quantum certificate of order $k$} is the matrix ${\Gamma^{\prime}}_{n+m}^k$, defined as 
\begin{eqnarray}
\nonumber ( {\Gamma^{\prime}}_{n+m}^k)_{ij}&=& \brakket{\Psi}{O_{i_1}^\dagger O_{i_2}^\dagger O_{j_2} O_{j_1}}{\Psi}\ ,
\end{eqnarray}
where $i=l\cdot (i_1-1)+i_2$, $j=l\cdot (j_1-1)+j_2$ and $l$ is the number of rows of a quantum certificate of order $k$ for the $n$-party quantum system. $O_{i_1}$ 
is the operator associated with the $i$\textsuperscript{th} row of the quantum certificate of order $k$ of the marginal $n$-party system (and similar for $O_{i_2}$ and 
the $m$-party system). 
\end{definition}
\begin{lemma}
 $ {\Gamma^{\prime}}_{n+m}^k\succeq 0$
\end{lemma}
\begin{proof}
 This follows directly form the fact that $ {\Gamma^{\prime}}_{n+m}^k$ is a sub-matrix  
 of the $(2k)$\textsuperscript{th} order 
  quantum certificate associated with the $(n+m)$-party quantum system.  
\end{proof}

The main insight, which will lead directly to the product theorems, is the following lemma.

\begin{lemma}\label{lemma:qproductconditions}
Let $P_{\bof{X}_1|\bof{U}_1}$ be an $n$-party and 
$P_{\bof{X}_2|\bof{U}_2}$ an $m$-party quantum system. Call the associated certificates of order $k$ $\Gamma_1^k$ and $\Gamma_2^k$ and write the linear conditions 
they fulfil 
as $A_{\mathrm{qb},1}\cdot \Gamma_1^k=0$, and 
$A_{\mathrm{qb},2}\cdot \Gamma_2^k=0$. Then the reduced quantum certificate of order $k$ associated with the $(n+m)$-party quantum system, fulfils
\begin{eqnarray}
\nonumber (A_{\mathrm{qb},1}\otimes \mathds{1}_{\Gamma_2^k})\cdot {\Gamma^{\prime}}_{n+m}^k=0\ \ \ \text{and}\ \ \ 
(\mathds{1}_{\Gamma_1^k} \otimes A_{\mathrm{qb},2} )\cdot {\Gamma^{\prime}}_{n+m}^k=0\ .
\end{eqnarray}
\end{lemma}
This can be interpreted the following way: even conditioned on any specific outcome of the second system, the first system must still be a quantum system. 
\begin{proof}
The matrix $A_{\mathrm{qb},1}$ contains entries of the form  $\brakket{\Psi}{O_{i_1}O_{j_1}}{\Psi}-\brakket{\Psi}{O_{i^{\prime}_1}O_{j^{\prime}_1}}{\Psi}=0$ which all operators 
associated with an $n$-party quantum system must fulfil, because $O_{i_1}O_{j_1}-O_{i^{\prime}_1}O_{j^{\prime}_1}=0$. 
By the definition of ${\Gamma^{\prime}}_{n+m}^k$, 
$(A_{\mathrm{qb},1}\otimes \mathds{1}_{\Gamma_2^k})\cdot {\Gamma^{\prime}}_{n+m}^k$ contains conditions of the form 
\begin{eqnarray}
\nonumber && \bra{\Psi}O_{i_1} O_{i_2} O_{j_2} O_{j_1}\ket{\Psi}- 
\bra{\Psi}O_{i^{\prime}_1} O_{i_2} O_{j_2} O_{j^{\prime}_1}\ket{\Psi}\\
\nonumber 
& =&
\bra{\Psi}O_{i_1} O_{j_1} O_{i_2} O_{j_2} \ket{\Psi}- 
\bra{\Psi}O_{i^{\prime}_1}  O_{j^{\prime}_1}O_{i_2} O_{j_2}\ket{\Psi}\\
\nonumber &=&
\bra{\Psi}(O_{i_1} O_{j_1} -O_{i^{\prime}_1}  O_{j^{\prime}_1}) O_{i_2} O_{j_2} \ket{\Psi}=0
\end{eqnarray}
where we have used the fact that operators associated with different parties commute, linearity, and the fact that the operators associated with an $(n+m)$-party 
quantum system must still fulfil the conditions 
 associated with a single system (as stated in 
 Definition~\ref{def:qbehavior}). 
\end{proof}

\subsection{A product lemma for the guessing probability}\label{subsec:productguess}

Using this property, we can show the product lemma (Theorem~\ref{th:guessprod}) for the guessing probability (for more details we refer to~\cite{thesis}). 
\begin{lemma}\label{lemma:qproduct}
 Consider the semi-definite program (\ref{eq:primalguess}) defined by $A_1,b_1,c_1$, bounding the guessing probability of $f(\bof{X}_1)$ of an $n$-party quantum system  
$P_{\bof{X}_1|\bof{U}_1}$, where $Q_1=(\bof{U}_1=\bof{u}_1,F=f)$. And similarly, associate 
 $A_2,b_2,c_2$ with an $m$-party quantum system $P_{\bof{X}_2|\bof{U}_2}$, where $g(\bof{X}_2)$ and $Q_2=(\bof{U}_2=\bof{u}_2,G=g)$. 
Then the guessing probability of $f(\bof{X}_1)\parallel g(\bof{X}_2)$ (denoting the concatenation) of the $(n+m)$-party system $P_{\bof{X}_1\bof{X}_2|\bof{U}_1\bof{U}_2}$ 
where $Q=(\bof{U}=\bof{u},F=f,G=g)$ is bounded by 
the semi-definite program $A,b,c$ with 
$b=b_1\otimes b_2$, $A=A_1\otimes A_2$.  
\end{lemma}
\begin{proof}
This follows form the fact that any $(n+m)$-party quantum system must fulfil Lemma~\ref{lemma:qproductconditions} and that $b_i\otimes b_j$ has a $1$ exactly at the 
entry associated with $\brakket{\psi}{O_1^\dagger O_2^\dagger O_2 O_1}{\psi}$, where $O_1$ is the operator associated with the probability of the outcome $\bof{x}_1$ 
mapped to a certain $f(\bof{x}_1)$, and similarly for $O_2$ and $g(\bof{x}_2)$. 
\end{proof}
Consider now the \emph{dual} of this `tensor product' problem. We will use a product theorem from~\cite{mittalszegedy} (see also~\cite{leemittal}) to show that 
for any dual feasible $\lambda$ (for a single system), $\lambda\otimes \cdots \otimes \lambda$ is dual feasible for the dual of the tensor product problem, therefore, 
forming an upper bound on the guessing probability. 

\begin{theorem}[Mittal, Szegedy~\cite{mittalszegedy}]\label{th:mittalszegedy}
Consider the semi-definite program $\min : c_1^T\cdot \lambda_1$, $\st : A_1^T\lambda_1-b_1\succeq 0$  and a feasible $\lambda_1$, and similarly for $A_2^T,b_2,c_2,\lambda_2$. 
Assume $b_1\succeq 0$ and $b_2 \succeq 0$.
Then  $\lambda=\lambda_1\otimes \lambda_2$ is feasible for the semi-definite program
$\min : (c_1\otimes c_2)^T\cdot \lambda$, $\st : (A_1\otimes A_2)^T\lambda-(b_1\otimes b_2)\succeq 0$. 
\end{theorem}
\begin{proof}
We use the fact that for a $\lambda$ such that $A^T\lambda-b\succeq 0$, where $b\succeq 0$, it holds that $A^T\lambda-b+2b=A^T\lambda+b\succeq 0$ because we 
consider a convex cone. The tensor product of two positive semi-definite matrices is positive semi-definite. We obtain
\begin{eqnarray} 
(A_1^T\lambda_1-b_1)\otimes (A_2^T\lambda_2+b_2)
\nonumber &=& 
A_1^T\lambda_1\otimes A_2^T\lambda_2 - b_1\otimes A_2^T\lambda_2 
+ A_1^T\lambda_1\otimes b_2 -b_1\otimes b_2\succeq 0\\
 (A_1^T\lambda_1+ b_1)\otimes (A_2^T\lambda_2- b_2)
\nonumber &=& 
A_1^T\lambda_1\otimes A_2^T\lambda_2 + b_1\otimes A_2^T\lambda_2 
- A_1^T\lambda_1\otimes b_2 -b_1\otimes b_2\succeq 0\ .
\end{eqnarray}
Adding the two inequalities and dividing by two, implies that 
\begin{eqnarray}
\nonumber A_1^T\lambda_1\otimes A_2^T\lambda_2-b_1\otimes b_2 = (A_1^T\otimes A_2^T)(\lambda_1\otimes \lambda_2)-b_1\otimes b_2 \succeq 0\ ,
\end{eqnarray}
which means that $\lambda_1\otimes \lambda_2$ is feasible for the product problem. 
\end{proof}

\begin{lemma}\label{lemma:qdualproduct}
Let $\lambda_1$ be a dual feasible solution (\ref{eq:dualguess}) defined by $A_1,b_1,c_1$ (see Lemma~\ref{lemma:qproduct}), and similarly for $\lambda_2$ and $A_2,b_2,c_2$. 
Then $\lambda=\lambda_1\otimes \lambda_2$ is dual feasible for $A,b,c$ where $A=A_1\otimes A_2$ and $b=b_1\otimes b_2$. 
\end{lemma}
\begin{proof}
Note that $b_i$ is of the form 
\begin{eqnarray}
\nonumber \left( \begin{array}{ccccc}
  0 & & \cdots & 0\\
0 & 1 & 0 \cdots 0\\
\cdots \\
0 & \cdots && 0
 \end{array}
\right)\ ,
\end{eqnarray}
i.e., it has a $1$ in the place where the matrix $\Gamma$ has the entry $\brakket{\Psi}{{E_u^x}^{\dagger}E_u^x}{\Psi}$ for $f(x)=i$ and $0$ everywhere else. 
It, therefore, only has positive entries on the diagonal and $0$ everywhere else.  
Clearly, $b_i\succeq 0$. 
The claim then follows by Theorem~\ref{th:mittalszegedy}. 
\end{proof}

We can now formulate the product lemma for the guessing probability. 
\begin{theorem}[Product lemma for the guessing probability]\label{th:guessprod}
Let $P_{\bof{X}_1|\bof{U}_1}$ be an $n$-party quantum system and $f(\bof{X}_1)$ a function $f:\mathcal{X}_1\rightarrow \mathcal{F}$ such that 
$P_{\mathrm{guess}}(f(\bof{X}_1)|Z(W),Q)\leq P_{\bof{X}_1|\bof{U}_1}^T\cdot \lambda_1$, where $Q=(\bof{U}_1=\bof{u}_1,F=f)$. Similarly, 
associate the guessing probability $P_{\mathrm{guess}}(g(\bof{X}_2)|Z(W,Q)\leq P_{\bof{X}_2|\bof{U}_2}^T\cdot \lambda_2$ with an $m$-party 
quantum system $P_{\bof{X}_2|\bof{U}_2}$ where
$Q=(\bof{U}_2=\bof{u}_2,G=g)$. Then the guessing probability of $f(\bof{X}_1)||g(\bof{X}_2)$ obtained from the $(n+m)$-party quantum system 
$P_{\bof{X}_1\bof{X}_2|\bof{U}_1\bof{U}_2}$ with 
$Q=(\bof{U}_1=\bof{u}_1,\bof{U}_2=\bof{u}_2,F=f,G=g)$
 is bounded by
\begin{eqnarray}
\nonumber P_{\mathrm{guess}}(f(\bof{X}_1)||g(\bof{X}_2)|Z(W),Q)&\leq &
P_{\bof{X}_1\bof{X}_2|\bof{U}_1\bof{U}_2}^T \cdot (\lambda_1 \otimes \lambda_2)\ .
\end{eqnarray}
\end{theorem}
\begin{proof}
This is a direct consequence of Lemma~\ref{lemma:qdualproduct}. 
\end{proof}
When the marginal system is of the form $P_{\bof{X}_1|\bof{U}_1}\otimes P_{\bof{X}_2|\bof{U}_2}$, this implies that the guessing probability is the product of the 
guessing probabilities of the two subsystems. In terms of the min-entropy, it implies that the min-entropy is additive. 

\begin{corollary}
Let $P_{\bof{X}|\bof{U}}=\bigotimes_n P_{{X}|{U}}$. Then 
\begin{equation}
\nonumber H_{\mathrm{min}}(\bof{X}|Z(W))=n\cdot  H_{\mathrm{min}}({X}|Z(W))
\end{equation} 
\end{corollary}

\section{Security under an Independence Assumption}\label{sec:qkd}

We have, in the previous sections, established all tools required for proving the security of quantum key distribution. The proof will consist of two steps. In the first, we will show that, using the 
above lemmas, we can have secure key distribution if the marginal distribution as seen by Alice and Bob looks like the product of several (identical) independent 
systems. In the next section, we will remove the condition of independence, because knowing that we are in a permutation invariant scenario, we will be able to relate 
the security of an arbitrary distribution to the security of independent distributions.

Roughly speaking, an entanglement-based quantum key distribution protocol proceeds along the following steps (we assume here, 
that Alice and Bob start with pre-distributed particle pairs described by a system $P_{XY|UV}^{\otimes n}$. 
\begin{itemize}
 \item Parameter estimation: Alice and Bob obtain a system $P_{XY|UV}^{\otimes n}$. In order to be able to bound Eve's knowledge about the raw key, they need 
to estimate the probability distribution $P_{XY|UV}$ of the individual systems. 
\item Information reconciliation: Alice sends some information about her raw key to Bob, such that he can correct the errors in his raw key. 
\item Privacy Amplification: Alice and Bob apply a public hash function to their raw keys in order to create a highly secure key. 
\end{itemize}

In the following we describe each of these steps in more detail and prove the technical results that will then 
constitute our security proof.

\subsection{Parameter estimation}\label{subsec:qpe}

Alice and Bob perform statistical tests on their system $P_{XY|UV}^{\otimes n}$ in order to estimate the probability 
distribution $P_{XY|UV}$ of the individual systems. They abort, if this distribution deviates from the desired one. 

The parameter-estimation protocol makes sure that only systems are accepted which have enough min-entropy, such that the final key will be secure. 
%
%
\begin{definition}
A parameter estimation protocol is said to \emph{$\epsilon$-securely filter} systems $P_{\bof{XY}|\bof{UV}}$ of a set $\mathcal{P}$ if on input  $P_{\bof{XY}|
\bof{UV}}\in \mathcal{P}$ the protocol outputs `abort' with probability at least $1-\epsilon$. 
It is said to be 
\emph{$\epsilon^{\prime}$-robust} on systems $P_{\bof{XY}|\bof{UV}}$ of a set $\mathcal{P}$ if on input  $P_{\bof{XY}|\bof{UV}}\in \mathcal{P}$ the protocol outputs 
`abort' with probability at most $\epsilon^{\prime}$.
\end{definition}

\begin{protocol}[Parameter estimation]\label{prot:qpe}\ 
\begin{enumerate}
\item Alice and Bob receive a system $P_{\bof{XY}|\bof{UV}}=P_{XY|UV}^{\otimes n}$. 
\item Alice chooses $\bof{u}$ such that for each $i$ with probability $1-k$, $u_i=\bar{u}_i$, where $\bar{u}$ denotes the input on which a raw key bit is generated. With probability $k$, she chooses $u_i$ uniformly at random amongst $\mathcal{U}$. 
\item Bob chooses $\bof{v}$ such that $v_i=\bar{v}_i$ with probability $1-k$ and 
with probability $k$, $v_i$ is chosen uniformly at random. 
\item They input $\bof{u}$ and $\bof{v}$ into the system and obtain 
the outputs $\bof{x}$ and $\bof{y}$. 
\item They exchange the inputs over the public authenticated channel. 
\item If less than $(1-k)^2 p  n$ inputs were $(\bar{u},\bar{v})$, they abort. 
\item Call $t$ the number of inputs where both chose not $\bar{u}$ and $\bar{v}$. 
If any combination $u,v$ occurred less than $k^2  p  n/|\mathcal{U}| |\mathcal{V}|$ times
they abort. 
\item From the inputs where they both chose a uniform input they estimate the distribution by $P^{\mathrm{est}}_{XYUV}(x,y,u,v)=\frac{1}{t}|\{i|(x_i,y_i,u_i,v_i)=(x,y,u,v)\}|$. 
Define ${\mathcal{P}}$ as the set of all $P_{XYUV}$ such that $|\mathcal{U}| |\mathcal{V}|\cdot P_{XYUV}^T\cdot \lambda \leq P_{\mathrm{guess}}$ for some dual feasible $\lambda$ 
(see (\ref{eq:dualguess})) and $P(X\neq Y|U=\bar{u},V=\bar{v})\leq \delta$. 
If $d(P_{XYUV}^{\mathrm{est}},P_{XYUV}^{\mathcal{P}})>\eta$ Alice and Bob abort, else, they accept. 
\end{enumerate}
\end{protocol}

We are now introducing some definitions which are used for the analysis of this protocol. 
\begin{definition}
Let $\mathcal{P}$ be a set of distributions $P_{XYUV}$. The set of systems  \emph{$\mathcal{P}^{\eta}$} are all distributions which have distance at least $\eta$ with 
the set $\mathcal{P}$, i.e., 
\begin{eqnarray}
\nonumber \mathcal{P}^{\eta} &=& \lbrace
 P_{XYUV}| d(P_{XYUV},P_{XYUV}^{\mathcal{P}}) > \eta \ \text{for all}\ P_{XYUV}^{\mathcal{P}}\in 
 \mathcal{P}
  \rbrace
\end{eqnarray}
\end{definition}

\begin{definition}
Let $\mathcal{P}$ be a set of distributions $P_{XYUV}$. The set of systems  \emph{$\mathcal{P}^{-\eta}$} are all distributions which have distance at least $\eta$ with 
the complement of the set $\mathcal{P}$, i.e., 
\begin{eqnarray}
\nonumber \mathcal{P}^{-\eta} &=& \lbrace
 P_{XYUV}| d(P_{XYUV},P_{XYUV}^{\bar{\mathcal{P}}}) > \eta \ \text{for all}\ P_{XYUV}^{\bar{\mathcal{P}}}\notin 
 \mathcal{P}
  \rbrace
\end{eqnarray}
\end{definition}

We further define the set of \emph{conditional systems} which are $\eta$-far or $\eta$-close to a certain set by the closeness of the distributions which can be obtained 
from them by choosing the input distribution to be uniform. 
\begin{definition}
Let $\mathcal{P}_{\mathrm{cond}}$ be a set of systems $P_{XY|UV}^{\mathcal{P}}$. For any system $P_{XY|UV}$, consider the distribution $P_{XYUV}=P_{XY|UV} /{|\mathcal{U}| 
|\mathcal{V}|}$. Then a system $P_{XY|UV}$ is in \emph{$\mathcal{P}^{\eta}_{\mathrm{cond}}$} if $P_{XYUV}\in  \mathcal{P}^{\eta}$
and $P_{XY|UV}$ is in \emph{$\mathcal{P}^{-\eta}_{\mathrm{cond}}$} if $P_{XYUV}\in  \mathcal{P}^{-\eta}$.
\end{definition}

Let us motivate, why we take exactly this definition of $\mathcal{P}^{\eta}_{\mathrm{cond}}$: the reason is, that it is useful to estimate $P_{XY|UV}^T \lambda$, 
where $P_{XY|UV}^T$ is the vector of all probabilities in the conditional distribution and $\lambda$ is some vector. This is in fact exactly the form of the bound on the 
guessing probability. 
\begin{lemma} \label{lemma:etaenvir}
Let $\mathcal{P}=P_{XY|UV}$.  
For all $P^{\bar{\eta}}_{XY|UV}\notin \mathcal{P}^{\eta}_{\mathrm{cond}} $, it holds that
\begin{eqnarray}
\nonumber {P^{\bar{\eta}}_{XY|UV}}^T \cdot \lambda - P_{XY|UV}^T \cdot \lambda & \leq & {P_{XY|UV}}^T \cdot  \lambda +|\mathcal{U}| |\mathcal{V}|\cdot  \eta \cdot 
\left(\sum_i|\lambda_i|\right). 
\end{eqnarray}
\end{lemma}
\begin{proof}
\begin{eqnarray}
\nonumber \left({P^{\bar{\eta}}_{XY|UV}}^T  - P_{XY|UV}^T \right)\cdot  \lambda &=&|\mathcal{U}| |\mathcal{V}|\cdot  {P^{\bar{\eta}}_{XYUV}}^T\lambda-|\mathcal{U}| |\mathcal{V}|\cdot  
{P_{XYUV}}^T\cdot  \lambda  \\
\nonumber 
&=& |\mathcal{U}|  |\mathcal{V}|\cdot \left({P^{\eta}_{XYUV}}^T-{P_{XYUV}}^T \right)\cdot \lambda\\
\nonumber & \leq & | \mathcal{U}|  |\mathcal{V}|\cdot  \eta \cdot \left(\sum_i|\lambda_i|\right).
\end{eqnarray}
\end{proof}

We will need the Sampling Lemma (Lemma~\ref{lemma:sampling}) to show that our protocol is secure, i.e., it $\epsilon$-securely filters input states with $\tilde{P}_{\mathrm{guess}}
\geq P_{\mathrm{guess}}+|\mathcal{U}| |\mathcal{V}|\cdot \eta \cdot \sum_i |\lambda_i|$ for the individual systems. 

\begin{lemma}{Sampling Lemma~\cite{KoeRen04b}}\label{lemma:sampling}
Let $Z$ be an $n$-tuple and $Z^{\prime}$ a $k$-tuple of random variables over a set $\mathcal{Z}$, with symmetric joint probability $P_{ZZ^{\prime}}$. Let $Q_{z^{\prime}}$ be the relative frequency 
distribution of a fixed sequence $z^{\prime}$ and $Q_{(z,z^{\prime})}$ be the relative frequency distribution of a sequence $(z,z^{\prime})$, drawn according to $P_{ZZ^{\prime}}$. Then for every $\ep\geq 0$ we 
have
\begin{eqnarray}
\nonumber  P_{ZZ^{\prime}}[||Q_{(z,z^{\prime})}-Q_{z^{\prime}}||\geq \ep]\leq |\mathcal{Z}|\cdot e^{-k\ep^2/8|\mathcal{Z}|}
\end{eqnarray}
\end{lemma}

\begin{lemma}\label{lemma:qfilters}
Protocol~\ref{prot:qpe} $\epsilon$-securely filters $\left(\mathcal{P}^{+\eta}_{\mathrm{cond}}\right)^{\otimes n}$ with 
\begin{eqnarray}
\nonumber \epsilon=
 |\mathcal{X}| |\mathcal{Y}||\mathcal{U}| |\mathcal{V}|\cdot e^{-\left(\frac{t^{\prime} \eta^2}{8 |\mathcal{X}|  |\mathcal{Y}|}\right)}\ , 
 \end{eqnarray} 
 where $t^{\prime}=k^2pn/|\mathcal{U}| |\mathcal{V}|$.
\end{lemma}
\begin{proof}
If for each of the conditional distributions $P_{XY|U=u,V=v}$ the estimate is within $\eta$, this also holds for the total distribution $P_{XYUV}$. By Lemma~\ref{lemma:sampling}, 
the probability that for any conditional distribution the estimate is $\eta$-far is at most $|\mathcal{X}||\mathcal{Y}|e^{-t^{\prime} \eta^2/8|\mathcal{X}||\mathcal{Y}|}$, where 
$t^{\prime}=k^2pn/|\mathcal{U}| |\mathcal{V}|$. We obtain the lemma by the union bound over all inputs. 
\end{proof}
Note that $\epsilon\in O(2^{-n})$ for any constant $0< k,p<1$ and $\eta>0$.

\begin{lemma}\label{lemma:peqrobust}
Protocol~\ref{prot:qpe} is $\epsilon^{\prime}$ robust on $\left( \mathcal{P^{-\eta}}\right)^{\otimes n}$  with 
\begin{eqnarray}
\nonumber \epsilon^{\prime}=
 |\mathcal{X}| |\mathcal{Y}||\mathcal{U}| |\mathcal{V}|\cdot e^{-\left(\frac{t^{\prime} \eta^2}{8 |\mathcal{X}|  |\mathcal{Y}|}\right)}
+e^{-2n \left((1-p)(1-k)^2 \right)^2}
+|\mathcal{U}||\mathcal{V}|\cdot e^{-2n \left(\frac{(1-p)k^2}{ |\mathcal{U}||\mathcal{V}|} \right)^2} \ , 
 \end{eqnarray} 
where $t^{\prime}=k^2pn/|\mathcal{U}| |\mathcal{V}|$.
\end{lemma}
\begin{proof}
This follows by the same argument as Lemma~\ref{lemma:qfilters} and a Chernoff bound (i.e., $\Prob[\frac{1}{n}\sum_i x_i \leq p-\ep]\leq e^{-2n\ep^2}$) on the probability 
that the protocol aborts because any of the inputs did not occur often enough. 
\end{proof}
It holds that $\epsilon^{\prime}\in O(2^{-n})$ for any constant $0< k,p<1$ and $\eta>0$.

\begin{lemma}
The protocol $\epsilon$-securely filters systems with $\tilde{P}_{\mathrm{guess}}\geq P_{\mathrm{guess}} + \eta^{\prime}$ for the individual system, where $\eta^{\prime}=|\mathcal{U}||\mathcal{V}|
\cdot \eta\cdot  \sum_i |\lambda_i|$.
\end{lemma}
\begin{proof}
This is a direct consequence of Lemma~\ref{lemma:qfilters} and Lemma~\ref{lemma:etaenvir}  and the fact that the guessing probability is given by $P_{XY|UV}^T \lambda$, 
see (\ref{eq:dualguess}). 
\end{proof}

\begin{lemma}
The protocol $\epsilon$-securely filters systems with $\tilde{\delta}\geq \delta + \eta^{\prime}$ for the individual systems. 
\end{lemma}
\begin{proof}
This follows from the definition of $\mathcal{P}^{+\eta}_{\mathrm{cond}}$. 
\end{proof}

\begin{lemma}
Assume the parameter estimation protocol $\epsilon$-securely filters inputs such that \linebreak[4] $\tilde{P}_{\mathrm{guess}}\geq P_{\mathrm{guess}} + \eta^{\prime}$ for the individual 
systems. Then it $\epsilon$-securely filters systems with 
$H_{\mathrm{min}}(X|E)_{\rho}\leq -n\log_2 \tilde{P}_{\mathrm{guess}}$. 
\end{lemma}
\begin{proof}
This follows from Theorem~\ref{th:krs} and the product lemma for the guessing probability (Lemma~\ref{lemma:qdualproduct}). 
\end{proof}

\subsection{Information reconciliation}\label{subsec:qir}

Having estimated the probability of error $\delta$ of their key bits in the previous section, Alice and Bob can do information reconciliation by applying a two-universal 
hash function\footnote{Information reconciliation using a two-universal hash function has the disadvantage, that the decoding procedure (i.e., for Bob to find $\bof{y}^{\prime}$) cannot 
be done in a computationally efficient way, in general. It is possible to use a \emph{code} for information reconciliation instead and there exist codes which can be 
efficiently decoded~\cite{holensteinphd}. However, in our setup the \emph{theoretical} efficiency of the decoding procedure is actually not important, as there exist codes 
with very good decoding properties in \emph{practice} and Alice and Bob can test whether they have correctly decoded using a short hash value of their strings. In case decoding 
does not succeed, they can repeat the protocol, resulting in some loss of robustness.} with output length $m$ bits, where $m=n\cdot h(\delta)+\ep$ and they can almost surely 
correct their errors, i.e., the keys will be equal apart from with exponentially small probability. 
\begin{definition}
Let $\mathcal{P}$ be a set of distributions $P_{\bof{XY}}$. An information reconciliation protocol is \emph{$\epsilon$-correct} on $\mathcal{P}$, if on input 
$P_{\bof{XY}}\in \mathcal{P}$  it outputs $\bof{x}^{\prime}$, $\bof{y}^{\prime}$ such that $\bof{x}^{\prime}\neq \bof{y}^{\prime}$ with probability at most $\epsilon$. 
It is \emph{$\epsilon^{\prime}$-robust} on $\mathcal{P}$, if on input $P_{\bof{XY}}\in \mathcal{P}$ it aborts with probability at most $\epsilon^{\prime}$. 
\end{definition}

\begin{protocol}[Information reconciliation]\label{prot:qir}\ 
\begin{enumerate}
\item Alice obtains $\bof{x}$ and Bob $\bof{y}$ distributed according to $P_{XY}^{\otimes n}$ with $\mathcal{X}=\mathcal{Y}=\{0,1\}$ and $P(X\neq Y)\leq \delta$. Alice outputs $\bof{x}^{\prime}=\bof{x}$. 
\item Alice chooses a function $f\in \mathcal{F}:\{0,1\}^n\rightarrow \{0,1\}^m$ at random, where $\mathcal{F}$ is a two-universal set of functions. 
\item She sends the function $f$ and  $f(\bof{x})$ to Bob.
\item Bob chooses $\bof{y}^{\prime}$ such that $d_H(\bof{y},\bof{y}^{\prime})$ is minimal among all strings $\bof{z}$ with $f(\bof{z})=f(\bof{x})$ (if there are two possibilities, he chooses 
one at random) and outputs $\bof{y}^{\prime}$. 
\end{enumerate}
\end{protocol}

The following theorem by Brassard and Salvail states that information reconciliation can be achieved by a two-universal function. We state the theorem with s slightly 
stronger bound on the error probability than the one originally given in~\cite{brassardsalvail}. 
\begin{theorem}[Information reconciliation~\cite{brassardsalvail}]\label{th:ir}
Let $\bof{x}$ be an $n$-bit string and $\bof{y}$ another $n$-bit string  obtained by sending $\bof{x}$ over a binary symmetric 
 channel with error parameter $\delta$. Assume the 
 function $f:\{0,1\}^n\rightarrow \{0,1\}^m$ is chosen at random amongst a set of two-universal functions.  
 Choose $\bof{y}^{\prime}$ such that $d_H(\bof{y},\bof{y}^{\prime})$ is minimal among all strings $\bof{r}$ with $f(\bof{r})=f(\bof{x})$. 
 Then, for any $\kappa>0$,  
\begin{eqnarray} 
\nonumber \Prob[{\bof{x}\neq \bof{y}^{\prime}}]\leq 
e^{-2\kappa ^2\cdot n}+
2^{n\cdot h(\delta+\kappa)-m} \ ,
\end{eqnarray}
 where $h(p)=-p\cdot \log_2 p - (1-p)\log_2 (1-p)$ is the binary entropy function.
\end{theorem}
\begin{proof}
$\bof{x}\neq \bof{y}^{\prime}$ if either $d_H(\bof{x},\bof{y})$ is large or if $f(\bof{x})=f(\bof{y}^{\prime})$. 
The probability that the strings $\bof{x}$ and $\bof{y}$ differ at more than $n(\delta+\kappa)$ positions is bounded by 
\begin{eqnarray}
\nonumber \Prob[d_H(\bof{x},\bof{y})]\geq n\cdot (\delta+\kappa)]\leq e^{-2\kappa^2\cdot n}\ .
\end{eqnarray}
The probability a $\bof{y}^{\prime}\neq \bof{x}$ with small $d_H(\bof{x},\bof{y}^{\prime})$ is mapped to the same value by $f$ is
\begin{eqnarray}
\nonumber \Prob[f(\bof{x})=f(\bof{y}^{\prime}),d_H(\bof{x},\bof{y}^{\prime})\leq n(\delta+\kappa) ] &\leq & 2^{-m}\cdot \sum_{i=0}^{ n(\delta+\kappa) }\binom{n}{i} \\
\nonumber &\leq & 2^{-m}2^{n\cdot h(\delta+\kappa)}\ .
\end{eqnarray}
The theorem follows by the union bound. 
\end{proof}

\begin{lemma}\label{lemma:qirworks}
The protocol is $\epsilon$-correct on input $P_{XY}^{\otimes n}$ such that $P(X\neq Y)\leq \delta$ where, for any $\kappa>0$,  
\begin{eqnarray}
\nonumber \epsilon &=& e^{-2\kappa^2\cdot n}+
2^{n\cdot h(\delta+\kappa)-m} \ . 
\end{eqnarray}
and $0$-robust on all inputs. 
\end{lemma}
\begin{proof}
Correctness follows directly from Theorem~\ref{th:ir}. Robustness follows from the fact that there always exists a $\bof{y}^{\prime}$ such that $f(\bof{y}^{\prime})=f(\bof{x})$.  
\end{proof}
For any $\kappa>0$ and $m>n\cdot h(\delta+\kappa)$, this value is $\in O(2^{-n})$.

When some information about the raw key is released --- such as, for example, when Alice and Bob do information reconciliation --- the min-entropy can at most be reduced 
by the number of bits com\-mu\-ni\-ca\-ted, see~\cite{rennerphd}. 
\begin{theorem}[Chain rule~\cite{rennerphd}]\label{th:chainrule}
Let $\rho_{XEC}$ be classical on $C$. Then 
\begin{eqnarray}
\nonumber H_{\mathrm{min}}(X|E,C)_{\rho}\geq H_{\mathrm{min}}(X|E)_{\rho}- H_{\mathrm{max}}(C)\geq H_{\mathrm{min}}(X|E)_{\rho}-m\ ,
\end{eqnarray}
where $m=\log_2 |C|$ is the number of bits of $C$. 
\end{theorem}

\subsection{Privacy amplification}\label{subsec:qpa}

In order to create a highly secure key from a partially secure string, Alice and Bob will do \emph{privacy amplification} i.e., apply a two-universal hash function to 
their raw keys. 
The distance from uniform of the final key string is given by the following theorem. 
\begin{theorem}[Privacy amplification~\cite{rennerkoenig,rennerphd}]\label{th:qpa}
Let $\rho_{XE}$ be classical on $\mathcal{H}_X$ and let $\mathcal{F}$ be a family of two-universal hash functions from $|\mathcal{X}|$ to $\{0,1\}^s$. Then
\begin{eqnarray}
\nonumber d(\rho_{F(X)EF}|EF)&\leq & \sqrt{\tr\rho_{XE}}\cdot 2^{-\frac{1}{2}(H_{\mathrm{min}(\rho_{XE}|E)}-s)}
 \leq 2^{-\frac{1}{2}(H_{\mathrm{min}(\rho_{XE}|E)}-s)}\ .
\end{eqnarray}
\end{theorem}

\subsection{Key distribution}\label{subsec:qkd}

We can now put everything together to obtain a key-distribution scheme. A key-distribution protocol should be \emph{secure}. This means that it should output the same 
key to Alice and Bob (\emph{correctness}) and Eve should not know anything about the key (\emph{secrecy}). Furthermore, the protocol should output a key when the adversary 
is passive, i.e., it should be \emph{robust}.

\begin{protocol}[Key distribution]\label{prot:qkey}\ 
\begin{enumerate}
\item Alice and Bob receive $P_{XY|UV}^{\otimes n}$
\item They apply parameter estimation using Protocol~\ref{prot:qpe}. 
\item They do information reconciliation using Protocol~\ref{prot:qir}. 
\item Privacy amplification: Alice chooses a function $f:\{0,1\}^n\rightarrow \{0,1\}^s\in \mathcal{F}$ from a two-universal set and sends $f$ to Bob. Alice outputs $f(\bof{x})$ 
and Bob $f(\bof{y}^{\prime})$. 
\end{enumerate}
\end{protocol}

\begin{lemma}\label{lemma:qkdsecure}
Protocol~\ref{prot:qkey} is $\epsilon$-secret with $\epsilon\in O(2^{-n})$ and $\epsilon^{\prime}$-correct with $\epsilon^{\prime}\in O(2^{-n})$ for $m>n\cdot h(\delta)$ and $s=q\cdot n< 
\log_2P_{\mathrm{guess}}-m/n$. It is $\epsilon^{\prime\prime}$-robust on $\left( \mathcal{P}^{-\eta}\right)^{\otimes n}$ with $\epsilon^{\prime\prime}\in O(2^{-n})$. 
\end{lemma}
\begin{proof}
This is a direct consequence of the fact that each step in the protocol is secure (Lemma~\ref{lemma:qfilters}, Lemma~\ref{lemma:qirworks} and Theorem~\ref{th:qpa}), 
taking into account Theorem~\ref{th:chainrule}. Robustness follows from the robustness of the parameter-estimation protocol, Lemma~\ref{lemma:peqrobust}. 
\end{proof}

The secret key rate is the length of the key $S$ that the protocol can output and still remain secure. We obtain the following.  
\begin{lemma}
The scheme reaches a key rate $q$ of 
\begin{eqnarray}
\nonumber q &=&- \log_2 P_{\mathrm{guess}}-h(\delta)
\end{eqnarray}
\end{lemma}

\begin{lemma}
The scheme reaches a positive key rate $q$ whenever 
\begin{eqnarray}
\nonumber -\log_2 P_{\mathrm{guess}}-h(\delta)&>& 0
\end{eqnarray}
\end{lemma}

\section{Removing the Assumption of Independence}\label{sec:notindependent}

We have seen that Alice and Bob can do key agreement (i.e., they either agree on a secret key or abort) if they share i.i.d. distributions. We now want to remove the 
requirement of independence. 

A special case is the one where Alice and Bob have two inputs and two outputs, i.e., their system violated the CHSH inequality~\cite{chsh}. In this case, there exists a 
(classical) map which they can apply to their inputs and outputs such that the system afterwards actually \emph{is} i.i.d.  more precisely a convex combination of i.i.d. 
distributions~\cite{mag,masanesv4}. The systems obtained this way, furthermore still violate the CHSH inequality by the same amount.\footnote{A similar map also exists for 
the generalization of the CHSH inequality, the Braunstein-Caves inequalities~\cite{braunsteincaves2}.}

In general, we do not know of such a map to transform arbitrary systems into product systems. Nevertheless, we will be able to relate the security of the key-distribution scheme 
on \emph{any} input to the security of the scheme on product inputs $P_{\bof{XY}|\bof{UV}}=P_{XY|UV}^{\otimes n}$ , for which we have already seen that it is secure, in 
Section~\ref{sec:qkd}. 
The reason is that we know that security is `permutation invariant' under the systems because each step of the protocol --- parameter estimation, information reconciliation 
and privacy amplification --- is permutation invariant\footnote{Otherwise permutation-invariance could be enforced by applying a random permutation on the systems at the beginning.}. 
The post-selection theorem allows us to relate security of permutation invariant states to the security of product states.

The post-selection theorem tells us that any permutation-invariant state can be obtained from the convex combination of i.i.d. (product) states by a measurement, and furthermore 
this measurement `works' sufficiently often. Therefore, if our key-distribution scheme is secure for product distributions, it is still `almost as secure' on a permutation 
invariant one.

Technically, the post-selection technique~\cite{postselection} gives a bound on the \emph{diamond norm} between two \emph{completely positive trace-preserving maps} (i.e., 
quantum channels) acting symmetrically on an $n$-party system. The diamond norm is directly related to the maximal probability of guessing whether one or the other map has 
been applied (on an input of choice), through the formula $p=1/2+1/4  \lVert \mathcal{E}-\mathcal{F} \rVert_\diamond$ (i.e., the distinguishing \emph{advantage} is then 
$1/4  \lVert \mathcal{E}-\mathcal{F} \rVert_\diamond$.) Therefore, it is especially useful in the context of cryptography, where usually a \emph{real} map is compared to an 
\emph{ideal} map --- such as one that creates a key that is secure by construction. While the diamond norm is defined as a maximization over all possible input states, the 
post-selection technique tells us that in the case of permutation invariant maps it is enough to consider them acting on a \emph{de Finetti state}, i.e., a convex combination 
of product states $\tau_{\mathcal{H}^n}=\int \sigma_{\mathcal{H}}^{\otimes n}\mu(\sigma_{\mathcal{H}})$, where $\mu$ is the measure induced by the Hilbert Schmidt metric. Let 
us now restate the main result of~\cite{postselection}. 
\begin{theorem}[Post-selection~\cite{postselection}]\label{th:postselection}
Consider a linear map from $\mathrm{End}(\mathcal{H}^{\otimes n})$ to $\mathrm{End}(\mathcal{H}^{\prime})$.\footnote{Note that in particular, $\Delta$ can be the difference between two 
completely positive trace-preserving maps $\mathcal{E}$ and $\mathcal{F}$.} 
If for any permutation $\pi$ there exists a CPTP map $\mathcal{K}_{\pi}$ such that $\Delta \circ \pi=\mathcal{K}_{\pi}\circ \Delta$, then 
\begin{eqnarray}
\nonumber \lVert \Delta \rVert_\diamond&\leq & g_{n,d} \lVert (\Delta \otimes \mathds{1}_{\mathcal{R}})\tau_{\mathcal{H}^n\mathcal{R}} \rVert_1\ ,
\end{eqnarray}
where $\mathds{1}_{\mathcal{R}}$ denotes the identity map on $\mathrm{End}(\mathcal{R})$ and $g_{n,d}=\binom{n+d^2-1}{n}\leq (n+1)^{d^2-1}$, where $d=\mathrm{dim}\mathcal{H}$. 
\end{theorem}

For our purposes, this means roughly 
\begin{eqnarray}
\nonumber  \Prob[\mathcal{E}(\sigma^{\pi})=\mathrm{insecure}]\leq (n+1)^{(d^2-1)} \int \Prob[\mathcal{E}(\sigma^{\otimes n})=\mathrm{insecure}]d\sigma\ ,
\end{eqnarray}
where $\sigma^{\pi}$ is a permutation invariant input and $\mathcal{E}$ denotes the event that the scheme is insecure. 
The very right-hand side is what we have analysed in the previous section and because this is exponentially small, it remains exponentially small even when multiplied by the 
polynomial factor in front of it.

In our case, $\sigma$ represents the system $P_{XY|UV}$. We, therefore, need to model $P_{XY|UV}$ by a quantum state (note that this is only a mathematical tool and does not have 
any physical meaning). More precisely, we represent the distribution $P_{XYUV}$ by $\sigma$. Since our parameter estimation protocol is such that it filters the conditional 
distribution independently of the input distribution (it aborts if any input does not occur often enough), this is equivalent.

\begin{lemma}
Let $P_{XYUV}$ be a probability distribution. Then there exists a density matrix $\sigma$ in a Hilbert space $\mathcal{H}$ with $dim(\mathcal{H}))=|\mathcal{X}| |\mathcal{Y}||
\mathcal{U}| |\mathcal{V}|$ such that measuring $\sigma$ in the standard basis gives the distribution $P_{XYUV}$.
\end{lemma}
\begin{proof}
Associate with each element of the standard basis $\{\ket{i}\}_i$ an outcome $x,y,u,v$. Take $\sigma =\sum_{i=1}^{|\mathcal{X}||\mathcal{Y}| |\mathcal{U}| |\mathcal{V}|}p_i \ket{i}
\bra{i}$ where $p_i=P_{XYUV}(x,y,u,v)$. 
\end{proof}
This tells us, that we can use $d=|\mathcal{X}||\mathcal{Y}||\mathcal{U}||\mathcal{V}|$ in the above formula. Let us now state, that the key-distribution protocol is secure on any 
input (not only product). It furthermore reaches essentially the same key rate. Robustness remains, of course, unchanged. 
\begin{theorem}\label{th:qsecure}
Protocol~\ref{prot:qkey} is $\epsilon$-secure with $\epsilon\in O(2^{-n})$ on any input for $m>n\cdot h(\delta)$ and $s=q\cdot n< \log_2P_{\mathrm{guess}}-m/n$. It is $\epsilon^{\prime\prime}$-robust
 on $\left( \mathcal{P}^{-\eta}\right)^{\otimes n}$ with $\epsilon^{\prime\prime}\in O(2^{-n})$. 
\end{theorem}
\begin{proof}
This follows directly from Lemma~\ref{lemma:qkdsecure}, using Theorem~\ref{th:postselection}.  
\end{proof}

\section{A Specific Protocol}\label{sec:qprotocol}

While our results apply to a rather generic class of protocols (see Protocol~\ref{prot:qkey}), we consider here, for the purpose of illustration, a specific protocol, as described below (Protocol~\ref{qprot}). The protocol is an entanglement-based quantum key distribution protocol similar to the original proposal by Ekert~\cite{ekert}.

\begin{figure}[ht!]
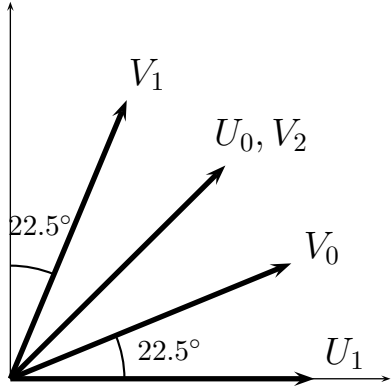

\centering
\pspicture*[](-0.5,-0.5)(5.5,5.5)
\psarc(0,0){1.5}{67.5}{90}
\rput[B]{0}(0.4,1.9){\normalsize{$22.5^{\circ}$}}
\psarc(0,0){1.5}{0}{22.5}
\rput[B]{0}(2.1,0.25){\normalsize{$22.5^{\circ}$}}
\rput[br]{90}(0,0){\psline[linewidth=0.5pt]{->}(0,0)(5,0)}
\rput[br]{0}(0,0){\psline[linewidth=0.5pt]{->}(0,0)(5,0)}
\rput[B]{0}(1.8,3.9){\Large{$V_1$}}
\rput[br]{67.5}(0,0){\psline[linewidth=2pt]{->}(0,0)(4,0)}
\rput[br]{22.5}(0,0){\psline[linewidth=2pt]{->}(0,0)(4,0)}
\rput[B]{0}(4.1,1.6){\Large{$V_0$}}
\rput[br]{45}(0,0){\psline[linewidth=2pt]{->}(0,0)(4,0)}
\rput[B]{0}(3.3,3.1){\Large{$U_0,V_2$}}
\rput[br]{0}(0,0){\psline[linewidth=2pt]{->}(0,0)(4,0)}
\rput[B]{0}(4.4,0.2){\Large{$U_1$}}
\endpspicture
\caption{\label{fig:basen} Alice's and Bob's measurement bases in terms of polarization used in Protocol~\ref{qprot}.}
\end{figure}

\begin{protocol}\label{qprot}\ 
\begin{enumerate}
 \item Alice creates $n$ maximally entangled states $\ket{\Psi^-}=(\ket{01}-\ket{10})/\sqrt{2}$, and sends one qubit of every state to Bob.
 \item Alice and Bob randomly measure the $i$\textsuperscript{th} system in either the basis $U_0$ or $U_1$ (for Alice) or $V_0$, $V_1$ or $V_2$ (Bob); the five bases are shown
 in Figure~\ref{fig:basen}. Bob flips his measurement result. 
They make sure that measurements associated with different subsystems commute. 
 \item The measurement results when both measured $U_0,V_2$ form the raw key.
\item For the remaining $k$ measurements they announce the results over the public authenticated channel and estimate the guessing probability and $\delta$ 
(see Section~\ref{subsec:qpe}). 
If the parameters are such that key agreement is possible, they continue; otherwise they abort.
\item They do information reconciliation and privacy amplification as  given in Sections~\ref{subsec:qir} and~\ref{subsec:qpa}. 
\end{enumerate}
\end{protocol}

\begin{figure}[h!]
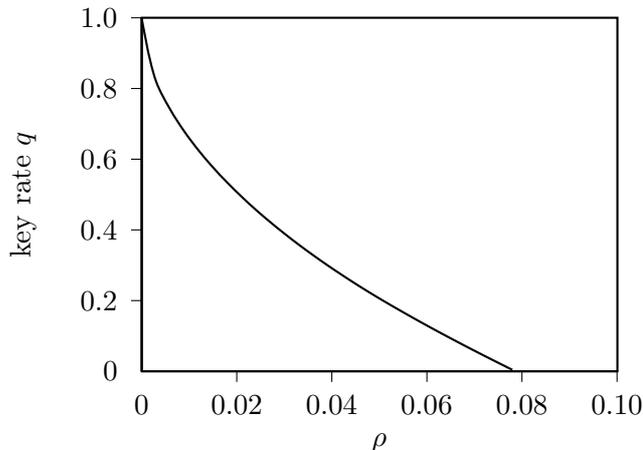

\centering
\pspicture[](-2.5,-0.75)(8.75,4.5)
 \psset{xunit=62.5cm,yunit=4.6875cm}
  \savedata{\mydataa}[
 {
{0.0000,1},
{0.0030,	0.821016426665049},
{0.0060,	0.740691751143548},
{0.0090,	0.677862617875497},
{0.0120,	0.624316166255036},
{0.0150,	0.57683385883745},
{0.0180,	0.533772043356894},
{0.0210,	0.494134634173186},
{0.0240,	0.457033494760889},
{0.0270,	0.422533310472383},
{0.0300,	0.389849475645926},
{0.0330,	0.35867529866217},
{0.0360,	0.329149447359709},
{0.0390,	0.300783150964238},
{0.0420,	0.273734569316761},
{0.0450,	0.247542826109787},
{0.0480,	0.222378305894379},
{0.0510,	0.198205818087642},
{0.0540,	0.174792660472844},
{0.0570,	0.151914854287686},
{0.0600,0.129554203443098},
{0.0630,0.107694014159461},
{0.066,0.0863189024482308},
{0.069,0.0652206073728114},
{0.072,0.0445824575238165},
{0.075,0.0242005174755597},
{0.078,0.00406681145525334}
} ]
  \rput[c](0.05,-0.2){{$\rho$}}
  \rput[c]{90}(-0.025,0.5){key rate {$q$}}
   \psaxes[Dx=0.02,Dy=0.2,  showorigin=true,tickstyle=bottom,axesstyle=frame](0,0)(0.1001,1.0001)
   \dataplot[plotstyle=curve,showpoints=false,dotstyle=o]{\mydataa} 
 \endpspicture
 \caption{\label{fig:qkeyrate} The key rate of Protocol~\ref{qprot} secure against device-independent quantum adversaries as function of the channel noise.}
 \end{figure}

Our main results (in particular Theorem~\ref{th:qsecure}) allow us to calculate the rate at which the protocol can produce a secure key, depending on the quality of the original entangled states shared by Alice and Bob. (This quality normally depends on the noise in the quantum channel used to distribute the entangled states.) For the matter of concreteness, we assume that these shared entangled states are mixtures consisting of a singlet (with weight $1-\rho$) and a fully mixed state (with weight $\rho$).  The resulting rate depending on the parameter $\rho$ is shown in Figure~\ref{fig:qkeyrate}.

\subparagraph*{Acknowledgements: }We thank Roger Colbeck, 
Dejan Dukaric, Artur Ekert, Thomas Holenstein, Severin Winkler and Stefan Wolf for helpful discussions. EH acknowledges support from 
the Swiss National Science Foundation and an ETHIIRA grant of ETH's research commission. RR acknowledges support from the Swiss National Science Foundation (grant No. 200021-119868).\\

\noindent \textit{Note added after completion of this work: }Results closely related to the ones presented here have been obtained independently in Ref.~\cite{masanespironioacin}.

\bibliographystyle{alpha}
\bibliography{quantum_pa_xor}

\appendix

\section{The XOR as Privacy-Amplification Function}

\subsection{Best attack on a bit}\label{subsec:qbit}

Of course, the analysis of Section~\ref{sec:qsingle} also tells us the best attack in case the function $f$ maps $\bof{X}$ to a bit. However, we can give a slightly different 
form to calculate the distance from uniform of a bit. This will allow us to show an XOR-Lemma for quantum secrecy. 
\begin{lemma}\label{lemma:distanceissdp}
Let $P_{\bof{X}Z|\bof{U}W}$ be a quantum system. 
The distance from uniform of 
$B=f(\bof{X})$ given $Z(W)$ and $Q:=(\bof{U}=\bof{u},F=f)$ is bounded by 
\begin{eqnarray}
\nonumber d(B|Z(W),Q)&\leq & \frac{1}{2}\cdot b^T\cdot \Gamma_\Delta^*\ ,
\end{eqnarray}
where $b^T\cdot \Gamma_\Delta^*$ is the optimal value of the optimization problem
\begin{eqnarray}
\label{eq:qprimal2} \max :&& \sum_{\bof{x}:B=0} \Gamma_\Delta(\bof{x},\bof{u})-\sum_{\bof{x}:B=1} \Gamma_\Delta(\bof{x},\bof{u})\\
\nonumber \st :&& A_{\mathrm{qb}}\Gamma_\Delta =0 
\\
\nonumber && \Gamma_\Delta \preceq \Gamma^k_{\mathrm{marg}} \\
\nonumber && \Gamma_\Delta\succeq-\Gamma^k_{\mathrm{marg}}\ ,
\end{eqnarray}
where $\Gamma^k_{\mathrm{marg}}$ is the matrix associated with the marginal system  
$P_{\bof{X}|\bof{U}}$. 
\end{lemma}
\begin{proof}
Define 
\begin{eqnarray}
\nonumber \Gamma_\Delta&=&2 p\cdot  \Gamma^{z_0}-\Gamma_{\mathrm{marg}}\ .
\end{eqnarray}
and note that with this definition $\Gamma^{z_0}=({\Gamma_{\mathrm{marg}}+\Gamma_\Delta})/({2p})$
and $\Gamma^{z_1}=({\Gamma_{\mathrm{marg}}-\Gamma_\Delta})/({2(1-p)})$. \\
The distance from uniform of a bit can be expressed as 
\begin{eqnarray}
\nonumber d(B|Z(W),Q)&=& \frac{1}{2}\cdot \left[p\cdot \left(\sum_{\bof{x}:B=0}\Gamma^{z_0}(\bof{x},\bof{u})-\sum_{\bof{x}:B=1}\Gamma^{z_0}(\bof{x},\bof{u})\right)\right.
\\
\nonumber
&&\left.+(1-p)\cdot \left(\sum_{\bof{x}:B=1}\Gamma^{z_1}(\bof{x},\bof{u})-\sum_{\bof{x}:B=0}\Gamma^{z_1}(\bof{x},\bof{u})\right)  
\right]\\
\nonumber&=& \frac{1}{2}\cdot b^T\cdot \Gamma_\Delta^*\ .
\end{eqnarray}
Now notice that $\Gamma^{z_0}$ and $\Gamma^{z_1}$ are actually quantum certificates of order $k$ exactly if $\Gamma_\Delta$ fulfils the above requirements. The conditions 
given by $A_{\mathrm{qb}}$ the matrix $\Gamma$ needs to fulfil are all linear and, therefore, because $\Gamma^k_{\mathrm{marg}}$ fulfils them, $\Gamma^{z_0}$ and $\Gamma^{z_1}$ fulfil them exactly if $\Gamma_\Delta$  does. 
The semi-definite constraints correspond exactly to the requirement that $\Gamma^{z_0}$ and $\Gamma^{z_0}$ are positive semi-definite, using the fact that the space of positive 
semi-definite matrices forms a convex cone. 
\end{proof}
The above semi-definite program can be written in the following form:
\begin{eqnarray}
\nonumber \text{PRIMAL}\\
\label{eq:primalbit} \max :& b^T\cdot \Gamma_\Delta\\
\nonumber \st :&
\underbrace{
 \left(
\begin{array}{c}
\phantom{-}\mathds{1} \\
-\mathds{1} \\
A_{\mathrm{qb}}
\end{array}
\right)
}_{A}
\cdot 
\Gamma_\Delta 
\begin{array}{c}
\preceq\\
\preceq\\
=
\end{array}
\underbrace{
\begin{array}{c}
\Gamma^k_{\mathrm{marg}}\\
\Gamma^k_{\mathrm{marg}}\\
0
\end{array}
}_{c}
\end{eqnarray}
\begin{eqnarray}
\nonumber  \text{DUAL}\\
\label{eq:dualbit} \min : & (\Gamma^k_{\mathrm{marg}})^T (\lambda_1+\lambda_2)\\
\nonumber  \st : & 
\underbrace{
\left(
\begin{array}{ccc}
\mathds{1} & -\mathds{1}& A_{\mathrm{qb}}^T 
\end{array}
\right)
}_{A^T}
\cdot 
\left(
\begin{array}{c}
\lambda_1 \\
\lambda_2\\
\lambda_3
\end{array}
\right)
=
b
\\
\nonumber  &\lambda_1,\lambda_2\succeq 0,\ \lambda_3 \text{ unrestricted}
\end{eqnarray}

\subsection[\ldots\ \  in terms of observable probabilities]{Best attack on a bit in terms of observable probabilities}\label{subsec:qbitobservable}

Any dual solution of (\ref{eq:dualbit}) gives us a bound on the distance from uniform of the bit $B$ in terms of the matrix elements $\Gamma^k_{\mathrm{marg}}$. We will now 
change our primal program to one where we \emph{optimize} over all $\Gamma^k_{\mathrm{marg}}$ compatible with the observable probabilities. The dual of this program has a solution 
only in terms these probabilities. We then show how we can transform any dual feasible solution of this program into a dual feasible solution of the program above with the same value.

The new semi-definite program we consider is the following:
\begin{eqnarray}
\label{eq:newopt}
\nonumber \text{PRIMAL}\\
\label{eq:qbitprimalobs} \max :& b^T\cdot \Gamma_\Delta\\
\nonumber \st :& \left(
\begin{array}{cc}
\phantom{-}\mathds{1} & -\mathds{1}\\
-\mathds{1} & -\mathds{1}\\
A_{\mathrm{qb}}& 0\\
0 & A_{\mathrm{IJ}}\\
0 & A_{\mathrm{qb}}
\end{array}
\right)
\cdot 
\left(
\begin{array}{c}
\Gamma_\Delta \\
\Gamma^k_{\mathrm{marg}}
\end{array}
\right)
\begin{array}{c}
\preceq\\
\preceq\\
=\\
=\\
=
\end{array}
\begin{array}{c}
0\\
0\\
0\\
P_{\bof{X}|\bof{U}}\\
0
\end{array}
\\
\nonumber  & \Gamma_\Delta, \Gamma^k_{\mathrm{marg}}\text{ unrestricted}
\end{eqnarray}
\begin{eqnarray}
\nonumber  \text{DUAL}\\
\label{eq:newoptdual} \min : & P_{\bof{X}|\bof{U}}^T \cdot \lambda_4\\
\nonumber  \st : & 
\left(
\begin{array}{ccccc}
\phantom{-}\mathds{1} & -\mathds{1} & A_{\mathrm{qb}}^T & 0 & 0\\
-\mathds{1} & -\mathds{1} & 0 & A_{\mathrm{IJ}} & A_{\mathrm{qb}}
\end{array}
\right)
\cdot 
\left(
\begin{array}{c}
\lambda_1 \\
\lambda_2\\
\lambda_3\\
\lambda_4\\
\lambda_5
\end{array}
\right)
=
\left(
\begin{array}{c}
b \\
0
\end{array}
\right)
\\
\nonumber  &\lambda_1,\lambda_2\succeq 0,\ \lambda_3, \lambda_4, \lambda_5 \text{ unrestricted}
\end{eqnarray}
where the matrix $A_{\mathrm{IJ}}$ is such that $A_{\mathrm{IJ}}\cdot \Gamma^k_{\mathrm{marg}}=P_{\bof{X}|\bof{U}}$. 
We claim that any dual feasible solution of (\ref{eq:newoptdual}) can be transformed into a dual feasible solution of (\ref{eq:dualbit}) with the same objective value. The 
solution of (\ref{eq:newoptdual}) therefore gives a bound on the distance from uniform only in terms of the observable probabilities. 
\begin{lemma}
Assume $ \lambda_1,\lambda_2,\lambda_3, \lambda_4, \lambda_5$ is a dual feasible solution of (\ref{eq:newoptdual}). Then $ \lambda_1,\lambda_2, \lambda_3$ is a dual feasible 
solutions of (\ref{eq:dualbit}) reaching the same objective value. 
\end{lemma}
\begin{proof}
The condition that  $\lambda_1,\lambda_2, \lambda_3$ is feasible for (\ref{eq:dualbit}) follows directly from the (upper row) feasibility condition of  (\ref{eq:newoptdual}). 
To see that it reaches the same value, we use that fact that $\Gamma^k_{\mathrm{marg}}$ is a quantum certificate, i.e., 
\begin{eqnarray}
\nonumber A_{\mathrm{qb}}\cdot \Gamma^k_{\mathrm{marg}}=0
\end{eqnarray}
and the (lower row) condition of (\ref{eq:newoptdual}), i.e.,
\begin{eqnarray}
\nonumber -\lambda_1-\lambda_2+A_{\mathrm{IJ}}^T\cdot\lambda_4+A_{\mathrm{qb}}^T\cdot\lambda_5=0\ .
\end{eqnarray}
We then obtain
\begin{eqnarray}
\nonumber {\Gamma^k}_{\mathrm{marg}}^T\cdot(\lambda_1+\lambda_2)&=&{\Gamma^k}_{\mathrm{marg}}^T\cdot(\lambda_1 +\lambda_2)
+\Gamma_{\mathrm{marg}}^T\cdot
( -\lambda_1-\lambda_2+A_{\mathrm{IJ}}^T\cdot\lambda_4+A_{\mathrm{qb}}^T\cdot\lambda_5)\\
\nonumber &=&{\Gamma^k}_{\mathrm{marg}}^T\cdot
(A_{\mathrm{IJ}}^T\cdot\lambda_4+A_{\mathrm{qb}}^T\cdot\lambda_5)\\
\nonumber &=& (A_{\mathrm{IJ}}\cdot {\Gamma^k}_{\mathrm{marg}})^T\lambda_4\\
\nonumber &=& P_{\bof{X}|\bof{U}}^T \cdot \lambda_4\ .
\end{eqnarray}
\end{proof}

\subsection{An XOR-Lemma for quantum secrecy}\label{subsec:qxor}

Using Lemma~\ref{lemma:qproductconditions}, we can now show that the XOR of the two partially secure bits is highly secure.

\begin{lemma}\label{lemma:qxorlemma}
Let $A_1,b_1,c_1$ be the parameters associated with the semi-definite program (\ref{eq:primalbit}) bounding the distance from uniform of a bit $f(\bof{X}_1)\in \{0,1\}$ obtained 
from an $n$-party quantum system $P_{\bof{X}_1|\bof{U}_1}$ where $Q=(\bof{U}_1=\bof{u}_1,F=f)$ Similarly 
associate $A_2,b_2,c_2$ with the distance from uniform of a bit $g(\bof{X}_2)\in \{0,1\}$ obtained from an $m$-party quantum system $P_{\bof{X}_2|\bof{U}_2}$. 
Then then the distance from uniform of the bit $f(\bof{X}_1)\oplus g(\bof{X}_2)$ obtained from the $(n+m)$-party system $P_{\bof{X}_1\bof{X}_2|\bof{U}_1\bof{U}_2}$, 
where $Q=(\bof{U}_1=\bof{u}_1,\bof{U}_2=\bof{u}_2,F=f,G=g)$ is bounded by the semi-definite program $A,b,c$ with $A=A_1\otimes A_2$ and $b=b_1\otimes b_2$.
\end{lemma}
\begin{proof}
This follows form the fact that any $(n+m)$-party quantum system must fulfil Lemma~\ref{lemma:qproductconditions} and $b$ describing the XOR of two bits can be described as the 
tensor product of the ones associated with each of the two bits. 
\end{proof}

This  implies that for any dual feasible solution, the tensor product is dual feasible for the tensor product problem.

\begin{lemma}\label{lemma:qdualxorproduct}
Let $\lambda_1$ be a dual feasible for (\ref{eq:dualbit}) with $A_1,b_1,c_1$ associated with an $n$-party quantum system and $\lambda_2$ dual feasible for an $m$-party quantum 
system described by $A_2,b_2,c_2$. Then $\lambda=\lambda_1\otimes \lambda_2$ is dual feasible for $A,b,c$ where $A=A_1\otimes A_2$ and $b=b_1\otimes b_2$.
\end{lemma}
\begin{proof}
 $\lambda_1 \otimes \lambda_2$ fulfils the dual constraints because 
\begin{eqnarray}
\nonumber [ A_1\otimes A_2](\lambda_1 \otimes \lambda_2)= b_1\otimes b_2 \ .
\end{eqnarray}
Furthermore, the tensor product of two positive semi-definite matrices is again positive semi-definite. 
\end{proof}

We can now formulate the XOR-Lemma for quantum secrecy. 
\begin{theorem}[XOR-Lemma for quantum secrecy]\label{th:xorquant}
Let $P_{\bof{X}_1|\bof{U}_1}$ be an $n$-party quantum system and $f(\bof{X}_1)$ a bit such that $d(f(\bof{X}_1)| Z(W),Q)\leq (P_{\bof{X}_1|\bof{U}_1}^T \lambda_1)/2$ 
with  $Q=(\bof{U}_1=\bof{u}_1,F=f)$. 
 Similarly, associate $d(g(\bof{X}_2)| Z(W),Q)\leq (P_{\bof{X}_2|\bof{U}_2}^T \lambda_2)/2$ with a bit from an $m$-party quantum system $P_{\bof{X}_2|\bof{U}_2}$ 
where
$Q=(\bof{U}_2=\bof{u}_2,G=g)$.
Then the distance from uniform of $f(\bof{X}_1)\oplus g(\bof{X}_2)$ obtained from the $(n+m)$-party quantum system $P_{\bof{X}_1\bof{X}_2|\bof{U}_1\bof{U}_2}$ with 
$Q=(\bof{U}_1=\bof{u}_1,\bof{U}_2=\bof{u}_2,F=f,G=g)$
 is bounded by
\begin{eqnarray}
\nonumber d(f(\bof{X}_1)\oplus g(\bof{X}_2)|Z(W),Q)&\leq & \frac{1}{2}
P_{\bof{X}_1\bof{X}_2|\bof{U}_1\bof{U}_2}^T (\lambda_1 \otimes \lambda_2)\ .
\end{eqnarray}
\end{theorem}
\begin{proof}
This follows directly from Lemma~\ref{lemma:qdualxorproduct}. 
\end{proof}

When the marginal system is of the form $P_{\bof{X}_1|\bof{U}_1}\otimes P_{\bof{X}_2|\bof{U}_2}$, this implies that the XOR is secure, whenever one of the two bits is secure.

\end{document}